\documentclass[copyright,creativecommons]{eptcs}
\providecommand{\event}{AiML 2026} 

\usepackage{aiml26}

\usepackage{iftex}

\ifpdf
  \usepackage{underscore}         
  \usepackage[T1]{fontenc}        
\else
  \usepackage{breakurl}           
\fi

\usepackage[T1]{fontenc}
\usepackage[misc,geometry]{ifsym}
\usepackage{stmaryrd}
\usepackage{amsmath}
\usepackage{amssymb}
\usepackage{xspace}
\usepackage{mathtools}
\usepackage{bussproofs}\EnableBpAbbreviations     
\usepackage{proof}
\usepackage{colortbl}
\usepackage{tikz}
\usetikzlibrary{shapes.geometric, arrows,arrows.meta}
\usepackage{float}
\usepackage{pifont}

\usepackage{xifthen}%
\usepackage{url}
\usepackage[inline]{enumitem}
\setlist{topsep=.3em,left=0pt}
\usepackage{xcolor}
\usepackage{array}
\usepackage{hyperref}
\usepackage{wrapfig}
\usepackage{mathdots}

\usepackage[color=green!20,linecolor=red]{todonotes}

\newcommand{\GL}{\mathbf{G}}
\newcommand{\GKL}{\mathbf{GK}}
\newcommand{\GWCL}{\mathbf{GW}^\mathrm{c}}
\newcommand{\GWL}{\mathbf{GW}}
\newcommand{\GCL}{\mathbf{G}^\mathrm{c}}
\newcommand{\GWCM}{\mbox{$\mathrm{GW}^\mathrm{c}$}}
\newcommand{\GCM}{\mbox{$\mathrm{G}^\mathrm{c}$}}
\newcommand{\GM}{\mbox{$\mathrm{G}$}}
\newcommand{\GWM}{\mbox{$\mathrm{GW}$}}
\newcommand{\KM}{\mbox{$\mathrm{K}$}}

\newcommand{\Calcw}{\Ccal_{\mathrm{GW}}}
\newcommand{\KL}{\mathbf{K}}

\newcommand{\IPL}{\mathbf{IPL}\xspace}

\newcommand{\ruleToLTAt}{\to_{\mathrm{At}}<}
\newcommand{\ruleToLT}{\to<}
\newcommand{\ruleToLEQAt}{\to_{\mathrm{At}}\leq}
\newcommand{\ruleToLEQ}{\to\leq}
\newcommand{\ruleToGAt}{\to_{\mathrm{At}}\gt}
\newcommand{\ruleToG}{\to\gt}

\newcommand{\tto}{\leftrightarrow}
\newcommand{\Diam}{\lozenge}

\newcommand{\Phibd}[1]{\Phi^{\Box,\Diam}(#1)}

\newcommand{\LC}{\Lcal_c}

\newcommand{\Qrange}{[0,1]_{\mathbf{Q}}}

\newcommand{\precm}{\prec_{\mathrm{m}}} 
\newcommand{\precc}{\prec_{\mathrm{c}}}

\newcommand{\PV}{\Vcal}

\newcommand{\Gat}{\G_\mathrm{at}}

\newcommand{\lab}[1]{\mathrm{label}(#1)}

\newcommand{\reass}[3]{#1[#2 \coloneqq #3]}

\newdimen\arrowsize 
\pgfarrowsdeclare{arcs'}{arcs'}{...} 
{ 
  \arrowsize=0.5pt 
  \advance\arrowsize by .5\pgflinewidth 
  \pgfsetdash{}{0pt} 
  \pgfsetroundjoin   
  \pgfsetroundcap    
  \pgfpathmoveto{\pgfpoint{-4\arrowsize}{4\arrowsize}} 
  \pgfpatharc{180}{270}{4\arrowsize} 
  \pgfpatharc{90}{180}{4\arrowsize} 
  \pgfusepathqstroke 
}

\makeatletter
\newbox\xrat@below
\newbox\xrat@above
\newcommand{\coimparrow}[2][]{%
  \setbox\xrat@below=\hbox{\ensuremath{\scriptstyle #1}}%
  \setbox\xrat@above=\hbox{\ensuremath{\scriptstyle #2}}%
  \pgfmathsetlengthmacro{\xrat@len}{max(\wd\xrat@below,\wd\xrat@above)+.65em}%
  \mathrel{\tikz [>-,>=arcs',baseline=-.5ex,line width=.6pt]
                 \draw (0,0) -- node[below=-2pt] {\box\xrat@below}
                                node[above=-2pt] {\box\xrat@above}
                       (\xrat@len,0) ;}}
\makeatother

\newcommand{\gt}{\rhd}
\newcommand{\lt}{\lhd}

\newcommand{\stru}[1]{\langle #1 \rangle} 
\newcommand{\abstractorder}{\,\blacktriangledown\,}

\newcommand{\provesw}[1]{\vdash_{\mathrm{GW}} #1}
\newcommand{\nprovesw}[1]{\nvdash_{\mathrm{GW}} #1}

\newcommand{\size}[1]{|#1|}
\newcommand{\sizem}[1]{||#1||}

\newcommand{\wgname}{\mathrm{wg}}
\newcommand{\wg}[1]{\mathrm{wg}(#1)}

\newcommand{\Atm}[1]{\mathrm{At}(#1)}
\newcommand{\Atmp}[1]{\mathrm{At}^{+}\!(#1)}

\newcommand{\Bs}{\mbox{\textsc{Bs}}}

 \newcommand{\Modname}{\mathrm{Mod}}
 \newcommand{\Mod}[1]{\Modname(#1)}
\newcommand{\app}[3]{#2\mapsto_{#1} #3}

\newcommand{\M}{\mathfrak{M}}

\newcommand{\ruleAxat}{\mathrm{Ax}\xspace}

\newcommand{\Search}{\textsc{Search}}
\newcommand{\Success}{\textsc{Success}}
\newcommand{\Fail}{\textsc{Fail}}

\renewcommand{\a}{\alpha}
\renewcommand{\b}{\beta}

\renewcommand{\phi}{\varphi}
\newcommand{\s}{\sigma}
\newcommand{\D}{\Delta}
\newcommand{\G}{\Gamma}

\newcommand{\Bcal}{\mathcal{B}}
\newcommand{\Ccal}{\mathcal{C}}

\newcommand{\Ical}{\mathcal{I}}

\newcommand{\Lcal}{\mathcal{L}}

\newcommand{\Qcal}{\mathcal{Q}}

\newcommand{\Tcal}{\mathcal{T}}
\newcommand{\Vcal}{\mathcal{V}}

\newcommand{\EndEx}{\mbox{}~\hfill$\Diamond$}

\newcommand{\medfont}{\fontsize{9pt}{10pt}\selectfont}

\RequirePackage{tcolorbox}
\definecolor{tcbx_Yellow_Bright}{RGB}{250, 250, 142}
\definecolor{tcbx_Yellow_Dark}{RGB}{230, 230, 0}
\newtcolorbox{tcyellowbox}[2][]{
  beforeafter skip=1\baselineskip,
  colbacktitle=tcbx_Yellow_Dark,
  top=0.2cm,
  bottom=0.2cm,
  left=0cm,
  colframe=tcbx_Yellow_Dark, 
  colback=tcbx_Yellow_Bright,
  halign lower=flush right,
  fonttitle=\bfseries,
  fontupper=\upshape,
  fontlower=\upshape,
  coltitle=black,
  title={#2},
  #1
}

\usepackage[linesnumbered,commentsnumbered,figure,noend]{algorithm2e}

\title{A G\"odel Modal Logic Over Witnessed Models}
\author{
  Mauro Ferrari$^{*}$
  \institute{Dep.~of Theoretical and Applied Sciences\\
    Universit\`a degli Studi dell'Insubria\\
    Varese, Italy.}
  \email{mauro.ferrari@uninsubria.it}
  \and
  Camillo Fiorentini
  \institute{
    Dep.~of Computer Science\\
    Universit\`a degli Studi di Milano\\
    Milano, Italy}
  \email{fiorentini@di.unimi.it}
  \and
  Paolo Giardini
  \institute{Dep.~of Theoretical and Applied Sciences\\
    Universit\`a degli Studi dell'Insubria\\
    Varese, Italy.}
  \email{pcgiardini@uninsubria.it}
  \and
  Ricardo Oscar Rodriguez
  \institute{
    UBA-FCEyN, Dep.~De Computación\\
    Buenos Aires, Argentina}
  \email{ricardo@dc.uba.ar}
}

\newcommand{\titlerunning}{A G\"odel Modal Logic Over Witnessed  Models}
\newcommand{\authorrunning}{M. Ferrari, C. Fiorentini, P. Giardini \& R.O. Rodriguez}

\begin{document}
\maketitle

\begin{abstract}
  We introduce $\GWL$, a G\"odel modal logic based on Kripke models in
  which the value of each modal formula is witnessed by an accessible
  world. This witnessed semantics eliminates the limit-based phenomena
  that preclude the finite model property in the usual Kripke
  semantics for G\"odel modal logics, thereby yielding a more
  constructive semantic framework. We present a sound and complete
  refutation calculus for $\GWL$ and design a terminating
  backward proof-search procedure with countermodel generation. As a
  direct consequence of this procedure, $\GWL$ enjoys the finite model
  property.
\end{abstract}

\section{Introduction}\label{sec:intro}

G\"odel modal logics, $\GKL$s, provide a natural framework for reasoning
about graded notions of incompleteness, uncertainty, and vagueness by
means of necessity and possibility operators evaluated over truth
degrees. When interpreted over the real unit interval, G\"odel logic
captures a well-behaved notion of implication based on residuation,
which makes it particularly suitable for applications in knowledge
representation, as well as in multi-agent and distributed systems.

In general, the semantics of $\GKL$ is based on G\"odel Kripke models
(\emph{\GM-models}), a natural generalization of classical Kripke
semantics in which both the valuation of propositions at each world
and the accessibility relation take values in the standard G\"odel
algebra $[0,1]$, i.e.\ the algebraic semantics of the intermediate
G\"odel--Dummett Logic. The minimal logics over {\GM-models} were
thoroughly investigated by Caicedo and Rodr\'iguez
in~\cite{CaiRod2010,CaiRod2015}.  In this setting, the interpretation
of the modal operators is defined by means of infima and suprema
computed over (possibly infinite) sets of accessible worlds. However,
these bounds need not be attained: the infimum may be strictly smaller
than every truth value in the set, and dually for the
supremum. Consequently, the truth value of a modal formula may depend
on limit constructions rather than on concrete worlds of the model.
This phenomenon has significant logical consequences. For instance,
\cite{CaiRod2010} shows that the formula $\Box \neg \neg \varphi \to
\neg \neg \Box \varphi$ is valid in all finite G\"odel Kripke models,
but fails in a suitable infinite model (see
Fig.~\ref{fig:nonWitn}). Hence, the underlying logic does not enjoy the
finite model property, and semantic arguments may require genuinely
infinitary behaviour.

Witnessed models provide a principled way to overcome these
difficulties. By requiring that the relevant infima and suprema be
realized by actual accessible worlds, witnessed semantics ensures that
existential and universal behaviour is supported by explicit
witnesses. In this way, the evaluation of modal operators more closely
mirrors the behaviour of guarded quantification in first-order logic
and restores a constructive reading of the semantics. Each modal
statement is thus justified by concrete worlds rather than by purely
asymptotic approximations.

The notion of witnessed models has appeared in several related
contexts. For example, in~\cite{HC2006} an axiomatization is given for
witnessed semantics in fuzzy predicate logics. Since quantified
statements can often be translated into modal ones, this connection is
conceptually well grounded. Along these lines, H\'ajek considers
in~\cite{HajekS5} the axioms $\Diamond(\varphi \to \Box \varphi)$ and
$\Diamond(\Diamond \varphi \to \varphi)$, and proves that the resulting
logic $S5(G)^w$ is strongly complete with respect to witnessed
universal Kripke models.

Witnessed semantics has also proved fruitful in applied settings. In
the context of Fuzzy Description Logics, \cite{BDP2014} studies a
system based on witnessed G\"odel semantics and shows that it enjoys
the finite model property and a decidable satisfiability problem with
exponential-time complexity. More recently, \cite{RV26} provides an
axiomatization for  G\"odel modal logic interpreted over witnessed
Kripke models.

In general, witnessed models tend to ensure compactness-like properties and
often yield the finite model property and decidability
results. These features are crucial for automated reasoning,
since they allow infinitary semantic arguments to be replaced by
finite combinatorial ones. For these reasons, studying G\"odel modal
logics over witnessed models not only clarifies their semantic
foundations, but also strengthens their connections with well-behaved
fragments of first-order logic that play a central role in computer
science, thereby enhancing both their theoretical robustness and their
applicability.

From a proof-theoretic perspective, witnessed semantics aligns more
closely with syntactic calculi.
The existence of witnesses enables canonical model constructions that
are finite or finitely generated, supports filtration techniques, and
allows completeness proofs to proceed by standard algebraic or
canonical methods. Therefore, establishing a sound and complete
calculus for G\"odel modal logics over witnessed models is not only
feasible but also conceptually justified, providing a proof-theoretic
counterpart to their improved semantic behaviour and paving the way for
effective automated reasoning procedures.

In line with the above considerations, in~\cite{FerFioRod:2025} we
introduced the logic $\GWCL$ characterized by witnessed crisp G\"odel
Kripke models, where ``crisp'' means that the accessibility relation
takes values in the set $\{0,1\}$.  For this logic, we developed a
constraint-based refutation calculus inspired by the tableau system
of~\cite{BilkovaFK:22}, together with a terminating proof-search
procedure that yields finite countermodels whenever proof search
fails. As a consequence, we established that $\GWCL$ enjoys the finite
model property.  In the present work, we extend these results to the
logic $\GWL$, characterized by witnessed G\"odel Kripke models without
the crispness restriction. Our development follows the
strategy adopted in~\cite{FerFioRod:2025}.  We introduce a
constraint-based refutation calculus for $\GWL$, called $\Calcw$, and
prove its soundness and completeness with respect to witnessed
models. Moreover, we design a terminating proof-search procedure for
$\Calcw$ based on a standard backward search strategy.  Notably, when
proof search fails, the computation can be traced to construct a
finite countermodel for the given formula, and this proves the finite
model property for $\GWL$.  Finally, we provide an implementation of
the procedure, called \texttt{gwref} (available
at~\cite{gwrefProver}).

\section{Basic Definitions}
\label{sec:basic}

Formulas are built over a countable set of propositional variables
$\PV$, using the propositional connectives $\land$, $\lor$, $\to$, the
constant $\bot$ and the modal connectives $\Box$, $\Diam$.  We introduce the following abbreviations:
$\neg \varphi$ stands  for $\varphi\to \bot$; 
$\phi_1\tto \phi_2$  stands  for $(\phi_1\to \phi_2)\land(\phi_2\to \phi_1)$.
The symbols
$\land$, $\lor$ and $\to$ are also used to denote algebraic
operations.
The \emph{G\"odel algebra} is a structure
$\stru{[0,1],\land,\lor,\to}$ where:
\[
a\land b = \min(a,b)
\qquad
a \lor b = \max(a,b)
\qquad
a\to b =
\begin{cases}
  1 & \mbox{if $a\leq b$}\\
  b &  \mbox{if $a > b$}     
\end{cases}
\]

\noindent
A \emph{\GM-model} (\emph{G\"odel model}) $\M$ is a structure
$\stru{W,R,e}$ where $W$ is a nonempty set (the set of worlds), $R$ is
a map $W\times W\to [0,1]$ (the accessibility relation) and $e : W
\times \PV\to [0,1]$ (the evaluation function).  If $R : W\times W\to
\{0,1\}$, we say that $\M$ is \emph{crisp}; in this case, $R \subseteq
W\times W$.  The map $e$ is extended to arbitrary formulas as follows:
\[
\begin{array}{ll}
  e(w,\bot) = 0 & e(w,\a \star \b) = e(w,\a) \star e(w,\b),\;\star\in \{\,\land,\,\lor,\,\to\,\}
  \\[1ex]
  e(w,\Box \a) = \inf_{w'\in W}  \{\,R(w,w') \to e(w',\a)  \,\}   \qquad
  &
  e(w,\Diam \a) =  \sup_{w'\in W}  \{\,R(w,w') \land e(w',\a)  \,\}.               
\end{array}
\]  

\begin{figure}[t]
  \centering
  \begin{minipage}{0.45\linewidth}\small
  \begin{tikzpicture}[scale=0.8]
      \draw[fill] (0,-0.5) circle (3pt)
      +(0,-0)   node (w0)  {}  
      +(0,-0.3) node  {\medfont$w_0$}
      ;
      
      \draw[fill]  (-4,1) circle (3pt)
      +(0,0)   node (w1)  {}  
      +(-0.6,0)  node  {\medfont $w_1$} 
      +(0,1) node  {
        $\begin{array}{l}
          p=\frac{1}{2},
          q=\frac{1}{2}\\[.5ex]
          \neg \neg p =1
        \end{array}$
      }
      ;
      
      \draw[fill] (-1,1) circle (3pt)
      +(0,0)   node (w2)  {}  
      +(-0.6,0) node{$w_2$}
      +(0,1) node{
        $\begin{array}{l}
          p=\frac{1}{3},
          q=\frac{2}{3}\\[.5ex]
          \neg \neg p =1
        \end{array}$
      }
      ;
      
      \draw[fill] (2,1) circle (3pt)
      +(0,0)   node (w3)  {}  
      +(-0.6,0) node{$w_3$}
      +(0,1) node{
        $\begin{array}{l}
          p=\frac{1}{4},
          q=\frac{3}{4}\\[.5ex]
          \neg \neg p =1
        \end{array}$
      }
      ;
      
      
      \draw[fill] (3,1) node{......}  ; 
      
      \draw[->] (w0) -- (w1) node[midway,above, xshift=-2em,yshift=-2ex] {1};
      \draw[->] (w0) -- (w2) node[midway, above] {1}  ;
      \draw[->] (w0) -- (w3) node[midway, above, xshift=0.5em,yshift=-2ex] {1} ;
    \end{tikzpicture}
  \end{minipage}\hspace{4ex}
  \fcolorbox{black}{gray!10}{
    \begin{minipage}{18em}\small
      \vspace{-2ex}
      \[
      \begin{array}{ll}
        W=\{\,w_j~|~j\geq 0\,\}&\\[1ex]
        R(w_0,w_0)=0&e(w_0,p)=e(w_0,q)=0\\[1ex]
        R(w_0,w_k)=1 \;\; k \geq 1
        &
        e(w_k,p) = \frac{1}{k+1}\;\; k\geq 1\\[1ex]
        &e(w_k,q) = \frac{k}{k+1}\;\; k\geq 1
      \end{array}
      \]
    \end{minipage}
  } 
  \small 
  $\small e(w_0,\Box p) \;=\; \inf_{j\geq 0}\{\, \mbox{\medfont $R(w_0,w_j)\to
  e(w_j,p) \,\} \;=\; \inf\{\; \overbrace{0\to 0}^1,\;\overbrace{1 \to
    (1/2)}^{1/2}, \; \overbrace{1 \to (1/3)}^{1/3}, \; \overbrace{1
    \to (1/4)}^{1/4}, \,\dots \;$}\}\;=\;0 $
  \\[1ex]
  $\small e(w_0,\Diamond q) \;=\; \sup_{j\geq 0}\{\, \mbox{\medfont $R(w_0,w_j)\land
  e(w_j,q) \,\} \;=\; \sup\{\; \overbrace{0\land 0}^0,\;\overbrace{1
    \land (1/2)}^{1/2}, \; \overbrace{1 \land (2/3)}^{2/3}, \;
  \overbrace{1 \land (3/4)}^{3/4}, \,\dots \;$}\}\;=\;1 $
  \\[2ex]
  $e(w_j,\neg \neg p) = 1,\; \forall j\geq 1\quad e(w_0,\Box\neg\neg
  p) =1 \quad e(w_0,\neg\neg \Box p) =0 \quad e(w_0,\ \Box\neg\neg
  p\to \neg\neg \Box p) =0 $
  \caption{A non-witnessed \GM-model $\M=\stru{W,R,e}$.}  
  \label{fig:nonWitn}
\end{figure}

\noindent
Note that $e(w,\neg\a)=1$ if $e(w,\a)=0$ and $e(w,\neg\a)=0$ if
$e(w,\a)>0$.  A world $w'$ is an $R$-successor of $w$ iff $R(w,
w')>0$. If $w$ has no $R$-successors, then $e(w,\Box\a)=1$ and
$e(w,\Diam\a)=0$ since, for every $w'\in W$, $R(w,w') \to e(w',\a) =
1$ and $R(w,w')\land e(w',\a) = 0$.  


A \GM-model $\M=\stru{W,R,e}$ is \emph{witnessed}, and is called a
\emph{\GWM-model}, if the values of $e(w,\Box\a)$ and $e(w,\Diam\a)$
are witnessed by some world $w'$ in $\M$, namely:
\begin{itemize}
\item $e(w,\Box \a)=r$ iff there is $w'\in W$ such that $R(w,w') \to
  e(w',\a) = r$;

\item $e(w,\Diam \a)=r$ iff there is $w'\in W$ such that
  $R(w,w')\land e(w',\a) = r$.
\end{itemize}
As a consequence, in the definitions of $e(w,\Box \a)$ and $e(w,\Diam
\a)$, the infimum and the supremum are just the minimum and the
maximum. A non-witnessed model is shown in Fig.~\ref{fig:nonWitn};
whenever the value $r$ of $R(w,w')$ or $e(w,p)$ is not displayed, it
is understood that $r$ is 0 (e.g., $R(w_i,w_j)=0$, for every $i\geq 1$
and $j\geq 0$).  It holds that $e(w_0,\Box p)=0$, but the witnessing
condition fails since there is no $w_j\in W$ such that $R(w_0,w_j) \to
e(w_j,p)=0$.  Similarly $e(w_0,\Diamond q)=1$, but there is no
witness.  When denoting a class of models, the superscript
$\mathrm{c}$ means ``crisp'' (e.g., $\GWCM$-model stands for ``crisp
$\GWM$-model'').

A formula $\varphi$ is \emph{valid} in $\M=\stru{W,R,e}$ iff
$e(w,\varphi)=1$ for every world $w$ in $W$.  If $\varphi$ is not
valid in $\M$, we say that $\M$ is a \emph{countermodel} for
$\varphi$; thus, $\M$ contains a world $w$ such that $e(w,\varphi) <
1$.  Let $\Qrange=[0,1]\cap\mathbf{Q}$, where $\mathbf{Q}$ is the set
of rational numbers; a model $\M=\stru{W,R,e}$ is \emph{discrete} if
$W$ is finite and the images of $R$ and $e$ are finite subsets of
$\Qrange$. We stress that every finite \GM-model is witnessed.

Let $\GL$ be the G\"odel-Dummett Logic, obtained by extending
Intuitionistic Propositional Logic $\IPL$ with the linearity axiom
$(\a\to\b)\lor (\b\to \a)$.
G\"odel Modal Logic $\GKL$ is obtained by
adding to $\GL$ the following axioms and rules ($\vdash$ refers to
provability in $\GKL$):
\[
\begin{tabular}{m{3em}m{12em}m{3em}m{16em}}
  $(K_\Box)$ & $\Box(\a \to \b)\to (\Box\a\to \Box \b)$ & 
  $(K_\Diam)$ & $\Diam(\a \lor \b)\to (\Diam\a\lor \Diam \b)$
  \qquad $(F_\Diam)\; \neg \Diam\bot$ 
  \\
  $(FS_1)$ & $\Diam(\a \to \b)\to (\Box\a\to \Diam \b)$ & 
  $(FS_2)$ & $(\Diam \a \to\Box \b)\to \Box(\a\to  \b)$
  \\
  $(N_\Box)$ & $\vdash \a$ implies $\vdash \Box\a$ &
  $(N_\Diam)$ & $\vdash \a\to \b$ implies $\vdash \Diam \a\to\Diam\b$
\end{tabular}
\]
In~\cite{CaiRod2015} it is proved that $\GKL$ is the set of formulas
valid in every \GM-model.  We introduce some logics extending $\GKL$,
by providing a semantic characterization
(see Fig.~\ref{fig:diagram}):

\begin{figure}[t]
  \centering
  \begin{minipage}{\linewidth}
  \begin{center}\small
  \begin{tikzpicture}[
    level 1/.style={sibling distance=10em},
    ]
    \node[fill=gray!10] (GWc) {$\GWCL$}
    child{ 
      node[fill=gray!10] (Gc) {$\GCL$}
      child[missing]{}
      child{
        node (GK)[fill=gray!10]  {$\GKL$}
        edge from parent
      } 
      edge from parent
    }
    child{ 
      node[fill=green!25] (GW) {$\GWL$}
      edge from parent
    } 
    ;
    \draw (GW) -- (GK) ;    

    \node[right=1ex of GWc] {\begin{minipage}{16ex}
      \small $(Cr)$, $(Wt)$  valid\end{minipage}};
    \node[right=1ex of GW] {\begin{minipage}{16ex}
      \small $(Cr)$ NOT valid \par$ (Wt)$ valid  \end{minipage}};
    \node[left=-1ex of Gc] {\begin{minipage}{16ex}
      \small   $(Cr)$  valid \par  $(Wt)$ NOT valid\end{minipage}};
  \node[right=1ex of GK] {\begin{minipage}{20ex}
      \small $(Cr)$, $(Wt)$   NOT valid\end{minipage}};

  \end{tikzpicture}
  \begin{minipage}{60ex}
    \begin{tabular}{p{30ex}p{50ex}}
      \begin{itemize}
      \item $(Cr)\quad \Box (\a \lor \b)\to (\Box \a \lor \Diamond \b)$
      \item $(Wt)\quad\Box \neg\neg \a \to \neg\neg\Box \a$
      \end{itemize}
      &
      \begin{itemize}
      \item $\GCL = \GKL + (Cr)$ (see \cite{RodriguezV:21})
      \item $\GKL\subset\GCL\subset\GWCL$,   $\GKL\subset\GWL\subset\GWCL$
      \end{itemize}
    \end{tabular}
  \end{minipage}
\end{center}
\end{minipage}
   \vspace{-2ex}
   \caption{Logics overview.}
  \label{fig:diagram}
\end{figure}

\begin{itemize}
\item $\GCL$ is the set of formulas valid in every  \GCM-model
  (see~\cite{RodriguezV:21});
\item $\GWL$ is the set of formulas valid in every \GWM-model;
\item $\GWCL$ is the set of formulas valid in every 
  \GWCM-model (see~\cite{FerFioRod:2025}).
\end{itemize}

\noindent
The logic $\GCL$ is obtained by adding the Crisp axiom
$\Box(\a\lor\b)\to (\Box \a \lor \Diam \b)$, which we call $(Cr)$, to
$\GKL$ (see~\cite{RodriguezV:21}).  The logic $\GWCL$ was introduced
in~\cite{FerFioRod:2025} and recently an axiomatization has been
proposed in~\cite{RV26}.  By definition,
$\GKL\subseteq\GCL\subseteq\GWCL$ and
$\GKL\subseteq\GWL\subseteq\GWCL$; we show that all the inclusions are
strict.  Let us consider the instance $\varphi =\Box (p\lor q)\to
(\Box p \lor \Diam q)$ of $(Cr)$ and let $\M=\stru{W,R,e}$ be the
$\GWM$-model in Fig.~\ref{fig:countGW}.  Since $e(w_0,\varphi) = 0.6$,
$\M$ is a countermodel for $\varphi$, and this certifies that
$\varphi\not\in\GWL$; accordingly, axiom $(Cr)$ is not valid in
$\GWL$.  Let $(Wt)$ be the formula schema $\Box \neg\neg \a\to
\neg\neg \Box \a$; we show that $(Wt)$ is valid in $\GWL$.  Let
$\M=\stru{W,R,e}$ be a $\GWM$-model and $w\in W$.  If $e(w, \neg\neg
\Box \a)=1$, we immediately get $e(w, (Wt))=1$. Otherwise, $e(w,
\neg\neg \Box \a)=0$, hence $e(w, \Box \a)=0$.  Since $\M$ is
witnessed, there exists $w'\in W$ such that $R(w,w') \to e(w',\a) =
0$; it follows that $e(w',\a)=0$, hence $e(w',\neg \neg \a)=0$.  This
implies $e(w, \Box\neg \neg \a)=0$, hence $e(w, (Wt)) = 1$. This
proves that $(Wt)$ is valid in $\GWL$.  Since every finite model is
witnessed, a countermodel for $(Wt)$ must be infinite.  The instance
$\Box \neg\neg p\to \neg\neg \Box p$ of $(Wt)$ is not valid in the
$\GCM$-model $\M$ in Fig.~\ref{fig:nonWitn} (see,
e.g.,~\cite{FerFioRod:2025}), thus $(Wt)$ is not valid in $\GCL$; this
also ascertains that neither $\GCL$ nor $\GKL$ enjoy the finite model
property.

We show that the (classical) modal logic $\KL$ can be embedded into
each of the logics in Fig.~\ref{fig:diagram}, by exploiting the
translation $*$ introduced in~\cite{MetOli1}.  We stress that a model
for the logic $\KL$ (a \emph{$\KM$-model}) can be viewed as a
$\GWCM$-model whose evaluation function takes values only in the set
$\{0,1\}$; accordingly, $\GWCL\subseteq\KL$.  We assume that $\Diam$
is not a primitive connective in the language of $\KL$ (actually, in
$\KL$ the formula $\Diam\varphi$ can be defined as
$\neg\Box\neg\varphi$).  Let $\a$ be a formula of $\KL$; the formula
$\a^*$ is obtained by replacing each propositional variable $p$
occurring in $\a$ with $\neg\neg p$.  As shown in the following lemma,
the evaluation of a formula $\a^*$ at any world of a $\GM$-model is
binary (0 or 1); this is essentially due to the fact that
propositional variables in $\a^*$ are double negated and $\Diam$ does
not occur in $\a^*$.

\begin{figure}[t]
  \centering
  \[\small
  \begin{array}{c}
    \begin{minipage}{10em}
      \begin{tikzpicture}[scale=0.8]
        \draw[fill] (0,0) circle (3pt)
        +(0,-0)   node (w0)  {}  
        +(0,-0.3) node  {$w_0$}
        ;
        
        \draw[fill]  (0,1.5) circle (3pt)
        +(0,0)   node (w1)  {}  
        +(-0.6,0)  node  {$w_1$} 
        +(1.2,0.2) node  {$p=0.5$}
        +(1.2,-0.2) node  {$q=0.6$}
        ;

        \draw[->] (w0) -- (w1) node[above right=-1.4] {{\small 0.6}} ;
      \end{tikzpicture}
    \end{minipage}
    \begin{minipage}{20em}
      \[\small
      \begin{array}{l}
        \fcolorbox{black}{gray!10}{
          \begin{minipage}{20em}\small
            $W=\{\,w_0,\,\,w_1\,\}$
            \\[.5ex]
            $R(w_0,w_0)=R(w_1,w_1)=0, \;R(w_0,w_1) = 0.6$
            \\[.5ex]
            $e(w_0,p) = e(w_0,q)=0$  
            \\[.5ex]
            $e(w_1,p) = 0.5,\; e(w_1,q)=0.6$
          \end{minipage}
        } 
      \end{array}
      \]
    \end{minipage} 
    \\[10ex]
    \begin{array}{rcl}
      e(w_0,\Box (p\lor q))&=&\inf\{\, R(w_0,w_0)\to e(w_0,p\lor q),\, R(w_0,w_1)\to e(w_1,p\lor q) \,\}\,=\,
      \inf\{\, 0\to 0,\, 0.6\to 0.6\,\}\,=\,1
      \\[1ex]
      e(w_0,\Box p) &=& \inf\{\, R(w_0,w_0)\to e(w_0,p),\,   R(w_0,w_1)\to e(w_1,p)\,\} \,=\,
      
      \inf\{\, 0\to 0,\,  0.6\to 0.5\,\}    \,=\, 0.5
      \\[1ex]
      e(w_0,\Diam q)&=& \sup\{\, R(w_0,w_0)\land e(w_0,q),\, R(w_0,w_1)\land e(w_1,q)\,\}\,=\,
      \sup\{\, 0\land 0,\,  0.6\land 0.6\,\}\,=\,  0.6
      \\[1ex]
      e(w_0,\varphi) &=& e(w_0, \Box (p\lor q))\to e(w_0, \Box p \lor \Diam q) \,=\, 1\to   0.6 \,=\, 0.6
    \end{array}
      
  \end{array}
  \]
  
  \caption{A countermodel for the formula $\varphi=\Box (p\lor q)\to
    (\Box p \lor \Diam q) $ (instance of (Cr)).}
  \label{fig:countGW}
\end{figure}

\begin{lemma}\label{lemma:K}
  Let $\a$ be a formula of  $\KL$. 

\begin{enumerate}[label=(\roman*), ref=(\roman*),leftmargin=*]    
\item\label{lemma:K:1}
If $\a\in\KL$ then $\a^*\in \GKL$.
  
\item\label{lemma:K:2}
If $\a^*\in \GWCL$ then  $\a\in\KL$. 
  
  \end{enumerate}
\end{lemma}

\begin{proof}
The  proof of point~\ref{lemma:K:2} is immediate, since
  $\GWCL\subseteq \KL$ and the formula  $\a\tto \a^*$  is valid in  $\KL$.
  We prove  point~\ref{lemma:K:1}. Let us assume that $\a^*\not\in \GKL$.
  There exists a $\GM$-model $\M=\stru{W,R,e}$ and  a world $w$ in $W$ such that $e(w,\a^*)< 1$.
Let $\M_1=\stru{W,R_1,e_1}$ be the   $\KM$-model  defined as follows:
\[
  R_1(w,w')\;=\;
  \begin{cases}
  1 &\text{if $R(w,w') > 0$}  
    \\
 0 & \text{otherwise}   
  \end{cases}
\hspace{4em}
  e_1(w,p)\;=\;
  \begin{cases}
  1 &\text{if $e(w,p) > 0$}  
    \\
 0 & \text{otherwise} 
\end{cases}
 \]  
 Let $w\in W$ and let $\varphi$ be any formula  of $\KL$; by induction on the size of $\varphi$, we prove that:
 \begin{enumerate}[label=(\Alph*), ref=(\Alph*),leftmargin=*]    
 \item\label{lemma:K:1:proof:1}
   $e_1(w,\varphi)= e(w,\varphi^*)$.
 \end{enumerate}
 The  case  $\varphi\in\PV\cup\{\bot\}$ is immediate.
  Let $\varphi = \Box \psi$;
note that $\varphi^*=\Box \psi^*$.
 Assume $e_1(w,  \Box \psi)=0$.
 There exists $w'\in W$ such that $R_1(w,w')=1$ and $e_1(w',   \psi)=0$.
 Then,  $R(w,w')>0$ and, by the induction hypothesis,   $e(w',   \psi^*)=0$.
 It follows that $R(w,w')\to e(w',   \psi^*)=0$, hence  $e(w,  \Box \psi^*)=0$.
 Let $e_1(w,  \Box \psi)=1$ and assume, by contradiction,  $e(w,  \Box \psi^*) <  1$.
 By the induction hypothesis,  for every $w' \in W$ it holds that  $e(w',\psi^*)\in\{0,1\}$,
 hence  $R(w,w')\to e(w',\psi^*) \in \{0,1\}$ and  $e(w,  \Box \psi^*) \in \{0, 1\}$.
 Thus,   $e(w,  \Box \psi^*) =  0$ and
 there is  $w'\in W$ such that  $R(w,w')\to e(w',\psi^*) = 0$.
 We get   $R(w,w') > 0$ and  $e(w',\psi^*)=0$,
 hence $R_1(w,w') = 1$ and  $e_1(w',  \psi)=0$,
 contradicting the assumption  $e_1(w,  \Box \psi)=1$.
Since the assumption $e(w,  \Box \psi^*) <  1$ yields a contradiction, 
 we conclude    $e(w,  \Box \psi^*) = 1$.
The proof of~\ref{lemma:K:1:proof:1} in the remaining cases easily follows from the induction hypothesis.
Since  $e(w,\a^*)< 1$, from~\ref{lemma:K:1:proof:1} it follows that
$e_1(w,\a) = 0$ . Accordingly, $\a\not\in\KL$,
and this proves~\ref{lemma:K:1}.
\end{proof}

As a consequence of Lemma~\ref{lemma:K}, we get:

\begin{proposition}\label{prop:embedding}
Let $\a$ be a formula of $\KL$ and let  $L$ be any of the logics $\GKL$, $\GCL$, $\GWL$ and $\GWCL$.   
Then, $\a\in\KL$ iff $\a^*\in L$.
\end{proposition}

We conclude this section by discussing the non-interdefinability of $\Box$ and $\Diam$.
This was established for  $\GCL$ in~\cite{RodriguezV:21}   using algebraic semantics.
Here, we instead employ  $\GWCM$-models  to extend the result to any modal logic contained in $\GWCL$.
We prove that:

\begin{enumerate}[label=(P\arabic*), ref=(P\arabic*),leftmargin=*]    
\item\label{defBD:P1}
  there is no  $\Box$-free formula $\varphi$ such that
  $\Box p \tto\varphi$ is valid in  $\GWCL$; 
  
\item\label{defBD:P2}
  there is no  $\Diam$-free formula $\psi$ such that
  $\Diam p \tto\psi$ is valid in  $\GWCL$. 
  
\end{enumerate}


\begin{lemma}\label{lemma:defBD}
  Let $\M^\ast=\stru{W,R,e}$ be the  $\GWCM$-model defined in Fig.~\ref{fig:ind}.
  \begin{enumerate}[label=(\roman*), ref=(\roman*),leftmargin=*]    
  \item\label{lemma:defBD:1} For every $\Box$-free formula $\varphi$,
    $e(w_0,\varphi)\in\{0,\,0.5,\,1\}$.
    
  \item\label{lemma:defBD:2} For every $\Diam$-free formula $\psi$,
    $e(w_0,\psi)\in\{0,\,0.4,\,1\}$.
  \end{enumerate}
\end{lemma}

\begin{proof}
  One can easily prove that, for every formula $\a$, the following
  facts hold:
  
  \begin{enumerate}[label=(\arabic*), ref=(\arabic*),leftmargin=*]    
  \item\label{lemma:defBD:F1}
    $e(w_1,\a)\in\{0,\,0.4,\,1\}$.
    
  \item\label{lemma:defBD:F2}
    $e(w_2,\a)=\Phi(e(w_1,\a))$,
    where $\Phi(0)=0$,  $\Phi(0.4)=0.5$, $\Phi(1)=1$.
  \end{enumerate}
  We prove~\ref{lemma:defBD:1}, by induction on the size of
  $\varphi$.  If $\varphi\in\PV\cup\{\bot\}$, we have $e(w_0,\varphi)
  = 0$.  The cases $\varphi=\a\land \b$, $\varphi=\a\lor \b$ and
  $\varphi=\a\to \b$ easily follow from the induction hypothesis.  Let
  $\varphi = \Diam \a$.  By~\ref{lemma:defBD:F1}
  and~\ref{lemma:defBD:F2}, one of the following
  properties~\ref{lemma:defBD:A}--\ref{lemma:defBD:C} holds:

  \begin{enumerate}[label=(\alph*), ref=(\alph*),leftmargin=*]    
  \item\label{lemma:defBD:A} $e(w_1,\a) = e(w_2,\a) = 0$.
    
  \item\label{lemma:defBD:B} $e(w_1,\a) = 0.4$ and $e(w_2,\a) = 0.5$.

  \item\label{lemma:defBD:C} $e(w_1,\a) = e(w_2,\a) = 1$.
  \end{enumerate}
  In case~\ref{lemma:defBD:A} we get $e(w_0,\Diam \a)=0$, in
  case~\ref{lemma:defBD:B} $e(w_0,\Diam \a)=0.5$, in
  case~\ref{lemma:defBD:C} $e(w_0,\Diam \a)=1$; this concludes the
  proof of~\ref{lemma:defBD:1}.  The proof of~\ref{lemma:defBD:2} is
  similar.
\end{proof}

\begin{proposition}\label{prop:gwclBD}
  Properties~\ref{defBD:P1} and~\ref{defBD:P2} hold.
\end{proposition}  

\begin{proof}
  Let us assume, by contradiction, that property~\ref{defBD:P1} does
  not hold.  Then, there exists a $\Box$-free formula $\varphi$ such
  that $\Box p \tto\varphi$ is valid in $\GWCL$. As a consequence, in
  the \GWCM-model $\M^\ast$ in Fig.~\ref{fig:ind} we have $e(w_0, \Box
  p \tto\varphi)=1$. Since $e(w_0, \Box p)= 0.4$, we get $e(w_0,
  \varphi)= 0.4$, in contradiction with
  Lemma~\ref{lemma:defBD}\ref{lemma:defBD:1}.  This proves
  that~\ref{defBD:P1} holds. The proof of~\ref{defBD:P2} is similar.
\end{proof}

\begin{proposition}\label{prop:defBD}
  $\Box$ and $\Diam$ are not interdefinable in any modal logic $L$
  contained in $\GWCL$.
\end{proposition}  

\begin{proof}
  Let $L\subseteq \GWCL$ and assume
  that  there is
  a $\Box$-free formula $\varphi$ such that $\Box p\tto\varphi\in
  L$. Then, $\Box p\tto\varphi\in \GWCL$,
  in contradiction with Lemma~\ref{prop:gwclBD}.  The proof for
  $\Diam$ is analogous.
\end{proof}

\begin{figure}[t]
  \centering
  \begin{minipage}{0.45\linewidth}
    \begin{tikzpicture}[scale=0.7]
      \draw[fill] (0,-0.5) circle (3pt)
      +(0,0)   node (w0)  {}  
      +(0,-0.3) node  {$w_0$}
      ;
            
      \draw[fill] (-2,1) circle (3pt)
      +(0,0)   node (w1)  {}  
      +(-0.6,0) node{$w_1$}
      +(0,0.5) node{$p=0.4$}
      ;
      
      \draw[fill] (2,1) circle (3pt)
      +(0,0)   node (w2)  {}  
      +(-0.6,0) node{$w_2$}
      +(0,0.5) node{$p=0.5$}
      ;
            
      \draw[->] (w0) -- (w1);
      \draw[->] (w0) -- (w2);
    \end{tikzpicture}
  \end{minipage}
  \begin{minipage}{0.5\linewidth}
    \[\small
    \begin{array}{l}
      \fcolorbox{black}{gray!10}{
        \begin{minipage}{16em}\small
          $W=\{\,w_0,\,w_1,\,w_2\,\}$
          \\
          $R =\{\,(w_0,w_1),\,(w_0,w_2)\,\}$
          \\
          $e(w_1,p) = 0.4,\,e(w_2,p) = 0.5$  
        \end{minipage}
      } 
    \end{array}
    \]
  \end{minipage}
  \caption{The \GWCM-model  $\M^\ast=\stru{W,R,e}$.}
  \label{fig:ind}
\end{figure}

\section{The Calculus $\Calcw$}
\label{sec_calculus}

\begin{figure}[t]
  \[
  \centering
  \begin{array}{c}
    \mbox{$\lt\in\{\,<,\,\leq\,\}$,\hspace{1em} $\gt \in \{\,>,\,\geq\,\}$}\qquad\qquad
    \AXC{}
    \RightLabel{\small$\ruleAxat$}
    \UIC{$\G$}    
    \DP
    \quad\mbox{if $\Atmp{\G}$ is not consistent}
    \\[3ex]
    \AXC{$w: \a \gt t,\, w: \b\gt t,\,  \G$}    
    \RightLabel{$\land\gt$}
    \UIC{$w: \a\land \b \gt t,\,\G$}
    \DP
    \hspace{4em}
    \AXC{$w: \a \lt t,\, \G$}    
    \AXC{$w: \b \lt t,\, \G$}
    \RightLabel{$\land\lt$}
    \BIC{$w: \a\land \b \lt t,\,\G$}
    \DP
    \\[4ex]    
    \AXC{$w: \a \lt t,\, w: \b \lt t,\,  \G$}    
    \RightLabel{$\lor\lt$}
    \UIC{$w: \a\lor \b \lt t,\,\G$}
    \DP
    \hspace{4em}
    \AXC{$w: \a \gt t,\,  \G$}
    \AXC{$w: \b \gt t,\,  \G$}
    \RightLabel{$\lor\gt$}
    \BIC{$w: \a\lor \b \gt t,\,\G$}
    \DP
    \\[4ex]
    \AXC{$w :\a >w : p,\, w : p  < t,\,  \G$}    
    \RightLabel{$\ruleToLTAt\;(\dag)$}
    \UIC{$w: \a\to p < t,\,\G$}
    \DP
    \hspace{1em}
    \AXC{$w :\a > w : q,\,  w :\b \leq w :q,\,  w : q  < t,\,  \G$}    
    \RightLabel{$\ruleToLT\;(\dag)$}
    \UIC{$w: \a\to \b < t,\,\G$}
    \DP
    \\[4ex]
    \AXC{$t \geq 1,\, \G$}    
    \AXC{$w :\a >w : p,\,w : p \leq  t,\,\G$}    
    \RightLabel{$\ruleToLEQAt\,(\dag)$}
    \BIC{$w: \a\to p \leq t,\,\G$}
    \DP
\\[4ex]    
    \AXC{$t \geq 1,\, \G$}    
    \AXC{$  w :\a > w : q,\,    w :\b \leq w : q,\,   w : q \leq  t,\,\G$}    
    \RightLabel{$\ruleToLEQ\,(\dag)$}
    \BIC{$w: \a\to \b \leq t,\,\G$}
    \DP
    \\[4ex]   
    \AXC{$w :\a \leq w : p,\,1\gt t,\, \G$}    
    \AXC{$w:p \gt t,\, \G$}
    \RightLabel{$\ruleToGAt\,(\dag)$}
    \BIC{$w: \a\to p \gt t,\,\G$}
    \DP
    \\[4ex]
    \AXC{$w :\a\ \leq w : q,\, w : \b \geq w :q ,\, 1\gt t,\, \G$}    
    \AXC{$w:\b \gt t,\, \G$}
    \RightLabel{$\ruleToG\,(\dag)$}
    \BIC{$w: \a\to \b \gt t,\,\G$}
    \DP
    \\[4ex]
    \begin{minipage}{1.0\linewidth}
      \[
      \begin{array}{l}
        \Atmp{\G} \;=\;
        \Atm{\G}\;\cup\;
        \{\, 1 > t~|~ w : \Box \a > t \in \G\,\}\;\cup\;
        \{\, 0 < t~|~ w : \Diam \a < t \in \G\,\}
        \\[1ex]   
        (\dag)\quad p\in\PV\cup\{\bot\},\;  \b\not\in\PV\cup\{\bot\},\;\mbox{$q$ is a new propositional variable}
        \end{array}
       \]
     \end{minipage}   
  \end{array}
  \]
  \vspace{-2ex}
 \caption{The calculus $\Calcw$, propositional rules ($\Gamma$ is a multiset of constraints).}
  \label{fig:calc1}
\end{figure}

\begin{figure}[t]
  \[
  \centering
  \begin{array}{c}
    \AXC{$R(w,w_1)\to w_1:\a \lt t,\;\Phibd{\G,w,w_1},\;\G$}
    \RightLabel{$\Box\lt$}
    \UIC{$w:\Box\a \lt t,\, \G$}    
    \DP
    \hspace{3em}
    \AXC{$R(w,w_1)\land w_1:\a \gt t,\;\Phibd{\G,w,w_1},\;  \G$}
    \RightLabel{$\Diam\gt$}
    \UIC{$w:\Diam\a \gt t,\,\G$}    
    \DP
    \\[2ex]
    \AXC{$R(w,w') \gt t,\, w': \a\gt t,\,  \G$}    
    \RightLabel{$R\land\gt$}
    \UIC{$R(w,w') \land w': \a \gt t,\,\G$}
    \DP
    \hspace{2em}
    \AXC{$R(w,w') \lt t,\, \G$}    
    \AXC{$w': \a \lt t,\, \G$}
    \RightLabel{$R\land\lt$}
    \BIC{$R(w,w') \land w':\a \lt t,\,\G$}
    \DP
    \\[4ex]
    \AXC{$w':\a<R(w,w'),\,  w':\a  < t,\,  \G$}    
    \RightLabel{$R\to <$}
    \UIC{$R(w,w') \to w':\a < t,\,\G$}
    \DP
    \hspace{2em}
    \AXC{$t \geq 1,\, \G$}    
    \AXC{$w':\a<R(w,w'),\,  w':\a \leq  t,\,\G$}    
    \RightLabel{$R\to\leq$}
    \BIC{$R(w,w')\to w':\a \leq t,\,\G$}
    \DP
    \\[4ex]    
    \AXC{$w' :\a \geq R(w,w'),\,1\gt t,\, \G$}    
    \AXC{$w':\a \gt t,\, \G$}
    \RightLabel{$R\to\gt$}
    \BIC{$R(w,w')\to w':\a \gt t,\,\G$}
    \DP
    \\[2ex]
     \begin{minipage}{1.0\linewidth}
       \[
       \begin{array}{l}
         \Phibd{\G,w, w_1}\,=\,
         \{\, R(w,w_1)\to w_1: \b \gt\!' t~|~w :\Box \b\gt\!' t \in\G  \,\}\,\cup\,
         \{\, R(w,w_1)\land w_1: \b \lt\!' t~|~w :\Diam \b\lt\!' t \in\G  \,\}\\[1ex]
         \mbox{$w_1$ is a new label}, \lt,\lt'\in\{<,\leq\}~\text{and}~\gt,\gt'\in\{>,\geq\}.
       \end{array}
       \]
     \end{minipage}
  \end{array}
  \]
  \vspace{-2ex}
  \caption{The calculus $\Calcw$: modal rules and $R$-rules ($\Gamma$ is a multiset of constraints).}
  \label{fig:calc2}
\end{figure}

Following~\cite{BilkovaFK:22,FerFioRod:2025}, we introduce the
refutation calculus $\Calcw$ for $\GWL$, namely, a calculus to certify
that a formula is not valid in any $\GWM$-model.  The calculus
exploits the constraint language $\LC$ defined over a countable set
of labels, used to represent the worlds of a \GWM-model.  In the
 definitions  below $w$, $w'$ are labels, $\varphi$  is a
formula, $p\in \PV$, $R$ is a designated relation symbol, representing
the accessibility relation.
\[
\begin{array}{rcl}
  \mbox{labelled formula} & \;\coloneqq&\; w :\varphi                            
  \\
  \mbox{atomic c-term $t$}  &\; \coloneqq&\; 0~|~1~|~w :p~|~w :\bot~|~R(w,w')
  \\[.5ex]
  \mbox{c-term $u$}  & \;\coloneqq&\; t~|~w : \varphi~|~R(w,w')\land w':\varphi~|~R(w,w')\to w':\varphi
  \\[.5ex]
  \mbox{constraint $\chi$}   &\; \coloneqq&\; u \abstractorder t \qquad \abstractorder\in \{\,<,\,\leq,\,>,\,\geq\,\}
\end{array}
\]
If  $w:\varphi$ occurs in   a c-term $u$, we say  that $u$ has label $w$.
A constraint of the form $t \abstractorder t$ is called \emph{atomic};
if $\chi= w:\varphi \abstractorder t$ is non-atomic and $\sharp$ is
the main connective of $\varphi$, we say that $\chi$ is a
\emph{$\sharp$-constraint}.
Constraints where $u$ is   $R(w,w')\land w':\varphi$ or
$R(w,w')\to w':\varphi$ are called
\emph{$R$-constraints}.
Given a multiset of constraints $\G$, by
$\Atm{\G}$ we denote the set of atomic constraints in $\G$.
Let $\M=\stru{W,R,e}$ be a \GWM-model. An \emph{$\M$-interpretation} of
$\LC$ is a function $\Ical$ mapping labels of $\LC$ to $W$.
We extend $\Ical$ to c-terms as follows:
\[
  \begin{array}{l}
    \Ical(k) = k,\;k\in\{0,1\}
    \qquad
    \Ical( w:\varphi) = e(\Ical(w), \varphi)
    \qquad
    \Ical\left(R(w,w')\right) =  R\left(\Ical(w),\Ical(w')\right)
    \\[1ex]
    \Ical\left(R(w,w') \star w':\varphi\right) =
   R\left(\Ical(w),\Ical(w')\right) \,\star\, e\left(\Ical(w'),\varphi\right),\;\star\in\{\land, \to\}
  \end{array}
\]  
Note that, for every c-term $u$, $\Ical(u)$ belongs to $[0,1]$.
 We introduce the relations
$\models_\Ical$ and $\models$, where $\M$ is a \GWM-model, $\Ical$ an
$\M$-interpretation, $\Gamma$ a multiset of constraints.
\begin{itemize}
\item $\M\models_\Ical u \abstractorder t$ iff   $\Ical(u)  \abstractorder \Ical(t)$;
  
\item $\M\models_\Ical \Gamma$ iff $\M\models_\Ical\chi$, for every
  $\chi\in\Gamma$;

\item $\M\models \Gamma$ iff $\M\models_\Ical \Gamma$ for some
  $\M$-interpretation $\Ical$.
\end{itemize}
A \emph{substitution} $\s$ is a function mapping each atomic c-term of
the form $w : p$ or $R(w,w')$ to a rational number in $\Qrange$; $\s$
is extended to all the atomic c-terms by setting $\s(k)=k$, for
$k\in\{0,1\}$, and $\s(w:\bot)=0$.  Let $\Gat$ be a set of atomic
constraints.  By $\s(\Gat)$ we denote the set of constraints obtained
by replacing every atomic c-term $t$ occurring in $\Gat$ with $\s(t)$.
Note that $\s(\Gat)$ is a set of rational constraints of the form $r_1
\abstractorder r_2$, with $r_1$ and $r_2$ in $\Qrange$; if all the
constraints in $\s(\Gat)$ hold, $\s$ is a \emph{solution} to $\Gat$.
The set $\Gat$ is \emph{consistent} iff it admits at least one
solution.  We remark that consistency of $\Gat$ can be checked by a
Constraint Solver over $Q$: one has to abstract the c-terms $w : p$
and $R(w,w')$ occurring in $\Gat$ by introducing new variables ranging
over $\Qrange$, and then check the consistency of the obtained
constraints by exploiting the solver.

Hereafter, in constraints $\lt$ and $\lt'$ are metasymbols ranging over
$\{<,\leq\}$; similarly $\gt,\gt' \in \{>,\geq\}$. The rules of the calculus
$\Calcw$ are displayed in Fig.~\ref{fig:calc1} (propositional rules)
and Fig.~\ref{fig:calc2} (modal rules and $R$-rules).
The \emph{main constraint} of a rule application is the constraint
displayed in the conclusion.  Rules for implication, having a main
constraint of the kind $w:\a\to \b\abstractorder t$, are defined
according to the structure of $\b$.  Let us consider the rule
$\ruleToLTAt$, having main constraint $w:\a\to p < t$
($p\in\PV\cup\{\bot\}$). The premise contains the constraints $w:\a >
w:p$ and $w:p < t$.  This reflects the fact that, given a $\GWM$-model
$\M$, a world $w$ in $\M$ and an $\M$-interpretation $\Ical$, if
$e(\Ical(w), \a\to p) < \Ical(t)$, it holds that $e(\Ical(w),\a) >
e(\Ical(w),p)$, $e(\Ical(w), \a\to p) = e(\Ical(w),p)$ and
$e(\Ical(w),p) < \Ical(t)$.  This rule cannot be generalized to any
$w:\a\to \b < t$ since the constraint $w:\a > w:\b$ is not allowed if
$\b\not\in\PV\cup\{\bot\}$.  This case is covered by rule $\ruleToLT$,
which introduces a new propositional variable $q$ behaving in $w$ as
$\b$, and $w:\a > w:\b$ is replaced by $w:\a > w: q$.  The
correspondence between $\b$ and $q$ is set by the constraint $w:\b\leq
w:q$, while the converse constraint $w:\b\geq w:q$ can be omitted.
There are two modal rules (see Fig.~\ref{fig:calc2}), the former
having main constraint $w:\Box\a \lt t$ and the latter $w:\Diam\a \gt
t$; both rules introduce a new label $w_1$.  $R$-rules (see
Fig.~\ref{fig:calc2}) are similar to the corresponding propositional
rules.  The definitions of \emph{tree} and \emph{derivation} of
$\Calcw$ are the usual ones (see, e.g.,~\cite{TroSch:00}).  We remark
that $\Calcw$ has the \emph{subformula property}, namely: if $\Tcal$
is a tree of $\Calcw$ having $\G_0$ as root, every formula occurring
in $\Tcal$ inside a constraint (e.g., the formula $\varphi$ in the
constraint $R(w,w')\to w':\varphi$) is a subformula of a formula in
$\G_0$ or a new propositional variable.  By $\provesw{\Gamma}$ we mean
that there exists a derivation of $\Gamma$.  In the rest of this
section we show that $\Calcw$ is a sound and complete refutation
calculus for $\GWL$. Formally, soundness and completeness are defined
as follows, where $\G$ is a multiset of constraints of the form
$w:\varphi\abstractorder t$:

\begin{itemize}
\item (Soundness) if $\provesw{\Gamma}$, then $\M\not\models\G$, for
  every \GWM-model $\M$;
\item (Completeness) if $\M\not\models\G$, for every \GWM-model  $\M$, then $\provesw{\Gamma}$.
\end{itemize}
We stress that, if   rule $\ruleAxat$   checked the
consistency of $\Atm{\G}$ instead of  $\Atmp{\G}$, the calculus $\Calcw$ would  not be complete.
Indeed, let  $\G=\{w : \Box p > w:p,\,w:p\leq 1,\,w:p\geq 1\}$.
Clearly, there is no  \GWM-model   $\M$  such that
$\M\models \G$ hence,  by completeness,  $\G$ must  be  provable in $\Calcw$.
To build a derivation of $\G$,  one can only exploit  rule $\ruleAxat$.
Note that $\Atm{\G}=\{w:p\leq 1,\,w:p\geq 1 \}$
is  consistent (any substitution $\s$ such that $\s(w:p)=1$ is a solution);
thus, if $\ruleAxat$  checked the consistency of  $\Atm{\G}$,
   $\G$ could not be proved.
Instead, $\G$ is proved  since $\ruleAxat$   evaluates
the set $\Atmp{\G} =\Atm{\G}\cup \{1 > w:p\}$, which  is not consistent.

An example of derivation is shown in Fig.~\ref{fig:der};
 the main constraint of a rule  application is underlined. 

\begin{figure}[t]
  \centering\small
  \begin{tabular}{p{\linewidth}}
    \[
    \begin{array}{ll}
      c=w_0:q_0\quad\Delta_0=\{\, w_0: \Box p > c,\; c \geq 1,\; c < 1\,\}&
      \Delta_1=\{\,  w_0:\bot \leq c,\;c < 1\,\}\\[.5ex]
      \Delta_2= \Delta_1\cup\{\, w_0: \Box p > c,\;  R(w_{0} ,w_{1}) > w_0:\bot\,\}&
      \Delta_3= \Delta_2\cup\{\,\;w_1: p\leq w_1:\bot,\; 1 > w_0:\bot   \,\}\\[.5ex]
      \multicolumn{2}{l}{\Delta_4= \Delta_2\cup\{\, R(w_{0} ,w_{1})\to w_1:p > c,\;w_1:\bot > w_0:\bot \,\}}\\
    \end{array}
    \]
    \\
    \infer[\to<]{
      \underline{w_0: \Box p \to  \neg  \Diamond  \neg p  < 1}}{
      \infer[\ruleToLEQAt]{
        w_0: \Box p > c,\;
        \underline{w_0: \neg  \Diamond  \neg p \leq c},\;   c < 1}{
        \infer[\textrm{Ax}]{
          \Delta_0}{
        }
        &
        \infer[\Diamond\gt]{
          w_0: \Box p > c,\;\underline{w_0: \Diamond  \neg p > w_0:\bot},\;\Delta_1}{
          \infer[R\land \gt]{
            w_0: \Box p > c,\; R(w_{0} ,w_{1})\to w_1:p > c,\; 
            \underline{R(w_{0} ,w_{1})\land w_1: \neg p > w_0:\bot},\;\Delta_1}{
            \infer[\ruleToGAt]{
              R(w_{0} ,w_{1})\to w_1:p > c,\;
              \underline{w_1: \neg p > w_0:\bot},\; \Delta_2}{
              \infer[R\to \gt]{
                \underline{R(w_{0} ,w_{1})\to w_1:p > c},\;\Delta_3}{
                \infer[\textrm{Ax}]{
                  w_1:p \geq R(w_{0} ,w_{1}),\;  1 > c,\;\D_3}{
                }
                &
                \infer[\textrm{Ax}]{w_1:p > c,\;\D_3}{
                }
              }
              &
              \infer[\textrm{Ax}]{\Delta_4}{}
            }
          }
        }
      }
    }
  \end{tabular}
  \vspace{-3ex}
  \caption{A derivation  of $w_0: \Box p \to  \neg  \Diamond  \neg p  < 1$.}
  \label{fig:der}
\end{figure}

Let $\M=\stru{W,R,e}$ be a $\GWM$-model and  $\Ical$  an $\M$-interpretation.
Let $p_0 \in \PV$ and $\a$  a formula not containing $p_0$, let $w_0\in\LC$  and   $w^\star\in W$.
The evaluation function   $\reass{e}{p_0}{\a}$   and the  $\M$-interpretation  $\reass{\Ical}{w_0}{w^\star}$
are defined as follows ($w\in W$ and $p\in\PV$):
\[
\reass{e}{p_0}{\a}(w,p)\:=\;
  \begin{cases}
  e(w,\a) &\text{if $p=p_0$}  
 \\[0.5ex]
 e(w,p) &\text{otherwise}   
  \end{cases}
\hspace{4em}
  \reass{\Ical}{w_0}{w^\star}(w) \:=\;
  \begin{cases}
  w^* &\text{if $w=w_0$}  
 \\[0.5ex]
 \Ical(w) &\text{otherwise}   
  \end{cases}
\]
Let $e'=\reass{e}{p_0}{\a}$ and $w\in W$; one can prove that, for
every formula $\varphi$ not containing $p_0$,
$e'(w,\varphi)=e(w,\varphi)$. Accordingly, $e'(w,p_0\tto \a)=1$, and
this implies that $\stru{W,R,e'}$ is a $\GWM$-model.  Let
$\Ical'=\reass{\Ical}{w_0}{w^\star}$ and let $\chi$ be a constraint
not containing $w_0$; one can prove that $\M\models_{\Ical'}\chi$ iff
$\M\models_{\Ical}\chi$.

\begin{lemma}\label{lemma:soundRule}
  Let $\rho$ be an instance of a rule of $\Calcw$ and 
 let $\Gamma$  be the conclusion of $\rho$. If there is  a \GWM-model $\M$
 such that $\M\models\G$, then there exist a
\GWM-model $\M'$ and a   premise $\G'$ of $\rho$
  such that $\M'\models\G'$.
\end{lemma}

\begin{proof}
  Let $\rho$ be an instance of the rule $\ruleAxat$; since the rule has no premises,
  we have to show that there is no $\M$ such that $\M\models \G$.  Let
  us assume, by contradiction, that there is a \GWM-model
  $\M=\stru{W,R,e}$ and an $\M$-interpretation $\Ical$ such that
  $\M\models_\Ical \G$.  If $ w : \Box \a > t \in \G$, then
  $e(\Ical(w), \Box \a) > \Ical(t)$, hence $1 > \Ical(t)$, which
  implies $\M\models_\Ical 1 > t$.  Similarly, if $ w : \Diam \a < t
  \in \G$, then $\M\models_\Ical 0 < t$.  This proves that
  $\M\models_\Ical \Atmp{\G}$.  By exploiting $R$ and $e$, one can
  define a solution to $\Atmp{\G}$;
  accordingly, $\Atmp{\G}$ is consistent, a contradiction.  We
  conclude that $\M$ does not exist. %

  Let $\rho\neq \ruleAxat$ and let us assume $\M\models_\Ical\G$,
  where $\M=\stru{W,R,e}$ is a \GWM-model and $\Ical$ is an
  $\M$-interpretation; we prove that there exist a \GWM-model $\M'$
  and a premise $\G'$ of $\rho$ such that $\M'\models\G'$.  We only
  detail some representative cases, by showing the corresponding rule
  application.
  \[\small
  \AXC{$\overbrace{w :\a > w : q,\,  w :\b \leq w :q,\,  w : q  < t,\,  \Delta}^{\G_1}$}    
  \RightLabel{$\to <$}
  \UIC{$\underbrace{w: \a\to \b < t,\,\Delta}_{\G}$}
  \DP
  \qquad
  \begin{minipage}{10em}
    $\b\not\in \Vcal\cup\{\bot\}$
    \\[.5ex]
    $q$ does not occur in $\G$
  \end{minipage}
  \]
  By the hypothesis,  $\M\models_{\Ical} w: \a\to \b < t$
  and $\M\models_{\Ical} \Delta$.  This implies that $e(\Ical(w), \a\to
  \b) < \Ical(t)$, hence $e(\Ical(w),\a) > e(\Ical(w),\b)$ and
  $e(\Ical(w),\a\to \b) = e(\Ical(w),\b)$, thus
$e(\Ical(w),\b) < \Ical(t)$. 
Let   $\M'=\stru{W,R,\reass{e}{q}{\b}}$;
we remark  that $\M'$ is a \GWM-model.
  Since $\M\models_{\Ical} \Delta$ and $q$ does not occur in
  $\Delta$, we get $\M'\models_{\Ical} \Delta$.  It is easy to check that
  $\M'\models_{\Ical} w :\b \leq w:q$ and $\M'\models_{\Ical} w :\a >
  w:q$ and $\M'\models_{\Ical} w:q < t$; we conclude $\M'\models_{\Ical}
  \G_1$.

  \[\small
  \AXC{$\overbrace{t \geq 1,\, \Delta}^{\G_1}$}    
  \AXC{$\overbrace{w :\a > w:q,\, w :\b \leq w:q,\, w:q \leq  t,\,\Delta}^{\G_2}$}    
  \RightLabel{$\to\leq$}
  \BIC{$\underbrace{w: \a\to \b \leq t,\,\Delta}_{\G}$}
  \DP
  \qquad 
  \begin{minipage}{10em}
    $\b\not\in \Vcal\cup\{\bot\}$
    \\[.5ex]
    $q$ does not occur in $\G$
  \end{minipage}
  \]
  By the hypothesis, $\M\models_{\Ical} w: \a\to \b \leq t$ and
  $\M\models_{\Ical} \Delta$, hence $e(\Ical(w), \a\to \b)\leq
  \Ical(t)$.  Let us assume $e(\Ical(w),\a) \leq e(\Ical(w),\b)$; then,
  $e(\Ical(w),\a\to \b) = 1$, which implies $1\leq \Ical(t)$, thus
  $\Ical(t) =1$.  We get $\M\models_\Ical t \geq 1$, hence
  $\M\models_\Ical \G_1$. Assume  $e(\Ical(w),\a) > e(\Ical(w),\b)$; then,  $e(\Ical(w),\a\to \b) = e(\Ical(w),\b)$,
  hence $\M\models_\Ical w:\b \leq t$.
  Let   $\M'=\stru{W,R,\reass{e}{q}{\b}}$;
 note that $\M'$ is a $\GWM$-model. 
  Since $\M\models_{\Ical} \Delta$ and $q$ does not occur in $\Delta$, we get
  $\M'\models_{\Ical} \Delta$.  Moreover, one can easily check that
  $\M'\models_{\Ical} w :\b \leq w:q$ and $\M'\models_{\Ical} w :\a >
  w:q$ and $\M'\models_{\Ical} w:q \leq t$; we conclude
  $\M'\models_{\Ical} \G_2$.

  \[\small
  \AXC{$\overbrace{w :\a \leq w:q,\,  w :\b \geq w:q,\,1\gt t,\, \Delta}^{\G_1}$}    
  \AXC{$\overbrace{w:\b \gt t,\, \Delta}^{\G_2}$}
  \RightLabel{$\to\gt$}
  \BIC{$\underbrace{w: \a\to \b \gt t,\,\Delta}_{\G}$}
  \DP
\qquad
  \begin{minipage}{10em}
    $\b\not\in \Vcal\cup\{\bot\}$
    \\[.5ex]
    $q$ does not occur in $\G$
  \end{minipage}
  \]
  By the hypothesis, $\M\models_{\Ical} w: \a\to \b \gt t$ and
  $\M\models_{\Ical} \Delta$, hence $e(\Ical(w),\a\to\b) \gt\Ical(t)$.
  Let us assume  $e(\Ical(w),\a) > e(\Ical(w),\b)$; then $e(\Ical(w),\a\to \b) =
  e(\Ical(w),\b)$, hence $\M\models_\Ical \G_2$.  Let us assume
  $e(\Ical(w),\a) \leq e(\Ical(w),\b)$; then,
  $e(\Ical(w),\a\to \b) = 1$, which implies $1 \gt \Ical(t)$, thus
  $\M\models_\Ical 1 \gt t$.  Let  $\M'=\stru{W,R,\reass{e}{q}{\a}}$;  note that  $\M'$ is a  \GWM-model. 
  Reasoning as in the previous cases,   we can prove that $\M'\models_{\Ical} \G_1$.
  \[\small
  \AXC{$\overbrace{R(w,w_1)\to w_1:\a \lt t,\,\Phibd{\G,w,w_1},\,\Delta}^{\G_1}$}
  \RightLabel{$\Box\lt$}
  \UIC{$\underbrace{w:\Box\a \lt t,\,\Delta}_{\G}$}    
  \DP
  \qquad
  \text{$w_1$ does not occur in $\G$}
\]
  By the hypothesis, $\M\models_{\Ical} w:\Box\a \lt t$ and
  $\M\models_{\Ical} \Delta$, hence $e(\Ical(w),\Box \a)\lt
  \Ical(t)$. 
  Since $\M$ is a $\GWM$-model, there exists a world $w^\star$ of $\M$
  such that $R(\Ical(w),w^\star)\to e(w^\star,\a)=e(\Ical(w),\Box\a)$.
  Let $\Ical'=\reass{\Ical}{w_1}{w^\star}$.
  We prove that:

  \begin{enumerate}[label=(\roman*), ref=(\roman*),leftmargin=*]    
  \item\label{prop:soundRule:box:1}     
 $\M\models_{\Ical'}\Delta$.

\item\label{prop:soundRule:box:2}     
$\M\models_{\Ical'} R(w,w_1)\to w_1:\a\lt t$.

\item\label{prop:soundRule:box:3}     
  $\M\models_{\Ical'} \Phibd{\G,w,w_1}$.
  \end{enumerate}
  Point~\ref{prop:soundRule:box:1} follows from the fact that $\M\models_{\Ical}\Delta$ and $w_1$ does not occur in $\D$.
  Since $e(\Ical(w),\Box \a)\lt \Ical(t)$ and $e(\Ical(w),\Box
    \a) = R(\Ical(w),w^\star)\to e(w^\star,\a)$, we get
    $R(\Ical(w),w^\star)\to e(w^\star,\a)\lt \Ical(t)$; this proves~\ref{prop:soundRule:box:2}.
    Let $\chi\in\Phibd{\G,w,w_1}$;  we show that  $\M\models_{\Ical'}\chi$,
    and this proves ~\ref{prop:soundRule:box:3}.
    Let  $\chi=R(w,w_1)\to w_1:\b\gt t'$. Then,  $w : \Box \b \gt t'\in \Delta$,
hence $e(\Ical(w),\Box \b) \gt \Ical(t')$. It follows that
 $R(\Ical(w),w^\star)\to e(w^\star,\b)\gt
 \Ical(t')$,  hence $\M\models_{\Ical'} \chi$.
 The  proof of   $\M\models_{\Ical'} \chi$ in the  case $\chi= R(w,w_1)\wedge w_1 : \b  \lt' t'$ is similar,
 taking into account that  $w : \Diam \b \lt' t'\in \Delta$.
From points~\ref{prop:soundRule:box:1}--\ref{prop:soundRule:box:3}, we conclude  $\M\models_{\Ical'} \G_1$.  
\end{proof}

\begin{proposition}[Soundness]\label{prop:sound}
  \begin{enumerate}[label=(\roman*), ref=(\roman*),leftmargin=*]    
  \item\label{prop:sound:1} If $\provesw{\G}$, then $\M\not\models
    \G$, for every \GWM-model $\M$.
    
  \item\label{prop:sound:2} $\provesw w:\varphi < 1$ implies
    $\varphi\in\GWL$.
  \end{enumerate}
\end{proposition}

\begin{proof} 
  Point~\ref{prop:sound:1} can be  proved by induction on the
  depth of a derivation of $\G$, by exploiting
  Lemma~\ref{lemma:soundRule}.  Let
  us  assume  $\varphi\not\in\GWL$.  Then, there exists a \GWM-model
  $\M=\stru{W,R,e}$ and a world $ w^\star\in W$ such that
  $e(w^\star,\varphi) < 1$.  Let $\Ical$ be an $\M$-interpretation
  mapping $w$ to $w^\star$; it holds that $\M\models_\Ical w:\varphi < 1$. 
  By~\ref{prop:sound:1},
we get $\nprovesw w:\varphi < 1$, and this
proves~\ref{prop:sound:2}. 
\end{proof}

\noindent
We show that $\Calcw$ is \emph{strongly terminating}, namely: there
exists a well-founded relation $\prec$ such that, for every
application $\rho$ of a rule of $\Calcw$, if $\G$ is the conclusion of
$\rho$ and $\G'$ any of the premises of $\rho$, then $\G'\prec \G$.
Equivalently, given a finite multiset $\G$ and repeatedly applying the
rules of $\Calcw$ upwards, proof search eventually halts, no matter
which strategy is used.

To prove strong termination of $\Calcw$, we define the relation
$\precc$ between multisets of constraints.  To this aim, in
Fig.~\ref{fig:weight} we introduce a weight function $\wgname$ on
constraints and multisets of constraints $\G$, together with the
multisets $\G[w]$ (multiset of constraints) and $\sizem{\G}$ (multiset
of natural numbers).

Given the finite multisets of natural numbers
$\Theta_1$ and $\Theta_2$,  we set:
\[
\Theta_1 \precm \Theta_2
\quad\mbox{iff}\quad
\Theta_1\neq
\Theta_2 \;\land\; \left(\; \forall k_1 \in \Theta_1\setminus
  \Theta_2.\ \exists k_2\in \Theta_2 \setminus \Theta_1.\ k_1 < k_2
  \;\right).
\]  
The relation $\precm$ is a multiset order, hence $\precm$ is
well-founded (see~\cite{BaaderN:98}, Th.~2.5.5 and Lemma 2.5.6).
Given the multisets of constraints $\G_1$ and $\G_2$,
we set $\G_1 \precc \G_2$ iff  $\sizem{\G_1} \precm\sizem{\G_2}$. 


The relation $\precc$ is well-founded;
in the next lemma we show that  $\precc$ 
can be used to instantiate $\prec$ in the definition of strong termination.

\begin{lemma}\label{lemma:rulesDec}
  Let $\rho$ be an instance of a rule of the calculus $\Calcw$, let
  $\Gamma$ be the conclusion of $\rho$ and  $\G'$ any of the premises of $\rho$.
  Then,  $\G'\precc \G$.
\end{lemma}

\begin{proof}
  We only discuss two representative cases.
  Let us consider the following application of rule $\to < $:
  \[
  \begin{array}{l}
    \AXC{$\overbrace{w :\a > w : q,\,  w :\b \leq w :q,\,  w : q  < t,\,\D_w,\,\Theta}^{\G_1}$}    
    \RightLabel{$\to < $}
    \UIC{$\underbrace{w: \a\to \b < t,\,\D_w,\,\Theta}_\G$}
 \DP
 \qquad
 \begin{minipage}{20em}
   all the constraints in $\Delta_w$ have label $w$
   \\[.5ex]
   no constraint in $\Theta$ has label $w$
 \end{minipage}
 \\[8ex]
 \G[w] =\{w:\a\to \b < t\} \cup \D_w 
 \hspace{3em}
 \G_1[w] =\{\,  w :\a > w : q,\,  w :\b \leq w :q,\,  w : q  < t  \,\}\cup\D_w
 \\[1ex]
 \sizem{\G} \,=\, \{\, \wg{\G[w]} \,\}\cup \sizem{\Theta}
 \hspace{4em}
 \sizem{\G_1} \;=\;
 \{\, \wg{\G_1[w]} \,\}    \cup \sizem{\Theta}  
\end{array}
\]
To prove that $\sizem{\G_1}\precm \sizem{\G}$,
we show that  $\wg{\G_1[w]} < \wg{\G[w]}$. 
The following holds:
\[
\begin{array}{rcl}
  \wg{\G[w]}&\,=\, &\wg{w:\a\to \b < t} + \wg{\D_w} \;=\; \wg{\a} +  \wg{\b} + 1 +  \wg{\D_w} 
  \\[1ex]
  \wg{\G_1[w]}&\,=\, & \wg{w :\a > w : q} +\wg{w :\b \leq w :q} + \wg{w : q  < t } + \wg{\D_w}
  \\
  &\,=\,&  \wg{\a} +  \wg{\b} + 0 +  \wg{\D_w} 
\end{array}
\]
Since $\wg{\G_1[w]} < \wg{\G[w]}$, we get $\sizem{\G_1}\precm
\sizem{\G}$, hence $\G_1\precc \G$.


\smallskip 

Let us consider the following application of $\Box\lt$:
\[
\begin{array}{l}
  \AXC{$\overbrace{\D_{w_1}, \;\D_w,\;\Theta}^{\G_1}$}
  \RightLabel{$\Box\lt$}
  \UIC{$\underbrace{w:\Box\a \lt t,\;\Delta_w,\;\Theta}_\G$}    
  \DP
  \qquad
  \begin{minipage}{22em}
    all the constraints in $\Delta_w$ have label $w$
    \\[.5ex]
    no constraint in $\Theta$ has label $w$
    \\[1ex]
    $\D_{w_1}\;=\;\{\,  R(w,w_1)\to w_1:\a \lt t\,\}\,\cup\, \Phibd{\Delta_w,w,w_1}$
  \end{minipage}
  \\[8ex]
  \G[w] =\{w:\Box\a \lt t\} \cup \D_w
  \hspace{5em}
  \G_1[w] =\D_w
  \hspace{5em}
  \G_1[w_1]=\D_{w_1}
  \\[1ex]
  \sizem{\G} \;=\; \{\, \wg{\G[w]} \,\}\,\cup\, \sizem{\Theta} 
  \hspace{4em}
  \sizem{\G_1} \;=\;
  \{\, \wg{\G_1[w]},\,   \wg{\G_1[w_1]} \,\}    \cup \sizem{\Theta}   
\end{array}
\]
We prove that $\sizem{\G_1}\precm   \sizem{\G}$. To this aim,  we show that:

\begin{enumerate}[label=(\roman*), ref=(\roman*),leftmargin=*]
\item\label{lemma:rulesDec:box1} $\wg{\G_1[w]} < \wg{\G[w]}$.
    
\item\label{lemma:rulesDec:box2}
  $\wg{\G_1[w_1]} < \wg{\G[w]}$.
\end{enumerate}
The proof of point~\ref{lemma:rulesDec:box1} is immediate since $\G[w]
= \{w:\Box\a \lt t\}\cup\G_1[w]$.  We
prove~\ref{lemma:rulesDec:box2}. We have:
\[
\begin{array}{l}
  \wg{\G[w]}\;=\; \wg{w:\Box\a \lt t } + \wg{\Delta_w} \;=\; \wg{\a} + 2 +  \wg{\Delta_w}
  \\[1ex]
  \wg{\G_1[w_1]}\;=\; 
  \wg{ R(w,w_1)\to w_1:\a } + \wg{ \Phibd{\Delta_w,w,w_1} }
  \;=\;  \wg{\a} + 1 + \wg{ \Phibd{\Delta_w,w,w_1}}
\end{array}
\]
To conclude the proof of~\ref{lemma:rulesDec:box2}, we show that:
  
\begin{enumerate}[label=(\roman*), ref=(\roman*),leftmargin=*,start=3]    
\item\label{lemma:rulesDec:box3}
  $\wg{ \Phibd{\Delta_w,w,w_1}} <  \wg{\Delta_w}$.
\end{enumerate}
Let $\chi_1\in \Phibd{\Delta_w,w,w_1}$. There exists $\chi\in
\Delta_w$ such that one of the two following conditions holds:
\begin{enumerate}[label=(\alph*), ref=(\alph*),leftmargin=*]    
\item\label{lemma:rulesDec:box3a}
  $\chi = w:\Box \b \gt' t'$ and $\chi_1 = R(w,w_1)\to w_1 : \b \gt' t'$;
  
\item\label{lemma:rulesDec:box3b}
  $\chi = w:\Diam \b \lt' t'$ and $\chi_1 = R(w,w_1)\land w_1 : \b \lt' t'$.
\end{enumerate}
In both cases, it holds that $\wg{\chi_1}= \wg{\b} +1$ and $\wg{\chi}=
\wg{\b} +2$, hence $\wg{\chi_1} < \wg{\chi}$; this
proves~\ref{lemma:rulesDec:box3}.  From~\ref{lemma:rulesDec:box1}
and~\ref{lemma:rulesDec:box2} it follows that $\sizem{\G_1}\precm
\sizem{\G}$, hence $\G_1\precc \G$.
\end{proof}

As an immediate consequence of Lemma~\ref{lemma:rulesDec}, we get:

\begin{proposition}\label{prop:term}
  The calculus $\Calcw$ is strongly terminating.
\end{proposition}

Consequently, any backward proof search strategy for $\Calcw$
terminates.  Following~\cite{FerFioFio:2013,FioFer:2021jlc}, we focus
on strategies where failure is certified by countermodels, i.e.: if
proof search for a multiset of constraints $\G$ fails, a \GWM-model
$\M$ such that $\M\models \G$ can be built; we call $\M$ a
\emph{countermodel} for $\G$, since $\M$ ascertains that $\G$ is not
provable in $\Calcw$ (by its soundness).  Let $\G$ be a multiset of
constraints:

\begin{itemize}
\item $\G$ is \emph{reduced} iff no rule of $\Calcw$ can be backward
  applied to $\G$; thus, every non-atomic constraint in $\G$ has the
  form $w:\Box \a \gt t$ or $w:\Diam \a \lt t$.

\item $\G$ is \emph{plain} iff only a modal rule can be backward
  applied to $\G$; thus, every non-atomic constraint in $\G$ has the
  form $w:\Box \a \abstractorder t$ or $w:\Diam \a \abstractorder t$.
\end{itemize}
An application of a modal rule is plain iff its conclusion is plain.
By $\app{\chi}{\G_k}{\G_{k+1}}$ we mean that $\G_{k+1}$ is a premise
of an application of a rule $\rho$ of $\Calcw$ having conclusion
$\G_k$ and main constraint $\chi$, with $\chi\in\G_k$.  A branch
$\Bcal$ is a sequence $\stru{\G_0,\dots,\G_n}$ such that, for every $0
\leq k < n$, there is $\chi_k$ such that
$\app{\chi_k}{\G_k}{\G_{k+1}}$; thus, $\Bcal$ represents a branch of a
tree of $\Calcw$.  The branch $\Bcal$ is \emph{saturated} iff the
following holds:

\begin{figure}[t]
  \centering
  \begin{tabular}{p{40ex}p{40ex}}\small
    \[
    \begin{array}{l}
      \wg{u} \;=\;
      \begin{cases}
        0 & \mbox{if $u$ is an atomic c-term}
        \\
        \size{\varphi}& \mbox{if $u = w:\varphi$}
        \\
        \size{\varphi} + 1&\begin{minipage}[t]{20ex}
          if $u = R(w,w')\star w':\varphi$\par~~~with $\star\in \{\land,\to\}$
        \end{minipage}
      \end{cases}
      \\[7ex]
      \wg{u\abstractorder t}\,=\,\wg{u}
      \hspace{3em}
      \wg{\G}\,=\,\sum_{\chi\in\Gamma} \wg{\chi}
      \\[1.5ex]
      \G[w] \,=\,\{\, u  \abstractorder t\in\G~|~\mbox{$\lab{u} = w$}\,\}
      \qquad
      \sizem{\G} \,=\, \{\, \wg{\G[w]}~|~\mbox{$w$ occurs in $\G$}    \,\}
    \end{array}    
    \]
    &\small
    \[
    \size{\varphi} \;=\;
    \begin{cases}
      0 & \mbox{if $\varphi\in\PV\cup\{\bot\}$}
      \\
      \size{\a} + \size{\b} + 1 &\begin{minipage}[t]{20ex}
        if $\varphi=\a\star\b$\par~~~with $\star\in \{\land,\lor,\to\}$
      \end{minipage}
      \\[2ex]
      \size{\a} + 2 & \mbox{if $\varphi=\Box\a$ or $\varphi=\Diam\a$}
    \end{cases}
    \]
  \end{tabular}
  \vspace{-2ex}
\caption{Weight function $\wgname$ on constraints and the multisets $\G[w]$,  $\sizem{\G}$.}
\label{fig:weight}
\end{figure}
  
\begin{itemize}[leftmargin=*]
\item $\G_0$ is a finite set of constraints of the form
  $w:\varphi\abstractorder t$ and  $\G_n$ is reduced;

\item every application of a modal rule in $\Bcal$ is plain.
\end{itemize}

\noindent
We show how to construct a countermodel for $\G_0$ from $\G_n$.  Let
$\Gat$ be a finite consistent set of atomic constraints and let $\s$
be a solution to $\Gat$.  By $\Mod{\Gat,\sigma}$ we denote the
$\GM$-model $\M=\stru{W,R,e}$ such that $W$ is the set of labels
occurring in $\Gat$ and:
\[\small
\begin{array}{l}
  R(w,w') \;=\;
  \begin{cases}
    \s(R(w,w')) & \mbox{if $R(w,w')$ occurs in $\Gat$}  
    \\
    0 & \mbox{otherwise}
  \end{cases}
  \qquad
  e(w,p) \;=\;
  \begin{cases}
    \s(w:p) & \mbox{if $w:p$ occurs in $\Gat$}  
    \\
    0 & \mbox{otherwise}
  \end{cases}
\end{array}
\]
We remark that $\M$ is a discrete $\GWM$-model.  A
$\Mod{\Gat,\sigma}$-interpretation $\Ical$ is \emph{canonical} iff
$\Ical(w)=w$ for every label $w$ occurring in $\Gat$. In
Lemma~\ref{lemma:count} we show that if $\Bcal=\stru{\G_0,\dots,\G_n}$
is saturated and $\s$ is a solution to $\Atmp{\G_n}$, then
$\Mod{\Atmp{\G_n},\s}$ is a countermodel for $\G_0$.

\begin{example}\label{ex:satBrancj}
  Let $\varphi=\Box(p\lor q) \to (\Box p\lor \Diam q)$ be an instance
  of $(Cr)$.  The tree shown in Fig.~\ref{fig:saturatedbranch}
  corresponds to a saturated branch $\Bcal=\stru{\G_0,\dots,\G_7}$,
  where $\G_0=\{w_0:\varphi < 1\}$ is the root and $\G_7$ is the top
  multiset.  When the applied rule has two premises, the annotation
  $(l)$ (left) or $(r)$ (right) specifies the selected one; in every
  rule application, the main constraint is underlined.

  The set $\Atmp{\G_7}$ consists of the following atomic constraints:
  \[\small
  w_1:p<R(w_0,w_1),\;\; w_1:p\leq c, \;\; w_1 : q \geq R(w_0,w_1),\;
  1> c,\;\; R(w_0,w_1)\leq c,\;\; c < 1
  \]
  The set $\Atmp{\G_7}$ is consistent; a solution to $\Atmp{\G_7}$ is
  any substitution $\s$ over $\Atmp{\G_7}$ such that
  \[
  \sigma(w_1:p)\, < \,\sigma\left(R(w_0,w_1)\right)\,\leq\, \sigma(c) \,< 1
  \qquad
  \sigma\left(R(w_0,w_1)\right)\,\leq \, \sigma(w_1: q) 
  \]
  For every solution $\s$, the \GWM-model $\M=\Mod{\Atmp{\G_7},\s}$
  satisfies $\M \models_\Ical w_0 :\varphi < 1$, where $\Ical$ is a
  canonical $\M$-interpretation; as a consequence, $\M$ is a countermodel for
  $\varphi$, witnessing that $\varphi\not\in\GWL$.
  Note that the worlds of $\M$ are $w_0$ and $w_1$ and,
  as expected, $\M$ is not crisp, since $0 < \sigma(R(w_0,w_1)) < 1$.
    A concrete solution
  $\s^\star$ is obtained by setting $\sigma^\star(w_1:p) = 0.5$ and
  $\sigma^\star(R(w_0,w_1))= \sigma^\star(c) = \sigma^\star(w_1:q) =
  0.6$.  The model $\Mod{\Atmp{\G_7},\s^\star}$ is essentially the
  same as the one displayed in Fig.~\ref{fig:countGW}.
  \EndEx
\end{example}

\begin{figure}[t]
  \centering
    \[\small
  \begin{array}{c}
    c = w_0:q_0\qquad\D\;=\;\{\,     w_0:  \Box(p\lor q) > c,\;\;  w_0: \Diam q  \leq  c,\;\;  c   < 1 \,\}    
    \\[2ex]
    \AXC{$w_1:p<R(w_0,w_1),\;\; w_1:p\leq c, \;\;  w_1 : q \geq R(w_0,w_1),\;\; 1> c,\;\;  R(w_0,w_1)\leq c,\;\;\D$}     
    \RightLabel{$R\land\lt(l)$}
    \UIC{$w_1:p<R(w_0,w_1),\;\; w_1:p\leq c, \;\;  w_1 : q \geq R(w_0,w_1),\;\; 1> c,\;\; \underline{R(w_0,w_1)\land w_1 : q \leq c},\;\;\D$}
    \RightLabel{$\lor\gt(r)$}
    \UIC{$w_1:p<R(w_0,w_1),\;\; w_1:p\leq c, \;\;  \underline{w_1 :p\lor q \geq R(w_0,w_1)},\;\; 1> c,\;\; R(w_0,w_1)\land w_1 : q \leq c,\;\;\D$}
    \RightLabel{$R\to\gt(l)$}
    \UIC{$w_1:p<R(w_0,w_1),\;\; w_1:p\leq c, \;\;  \underline{R(w_0,w_1)\to w_1 :p\lor q > c},\;\; R(w_0,w_1)\land w_1 : q \leq c,\;\;\D$} 
    \RightLabel{$R\to\leq(r)$}
    \UIC{$\underline{R(w_0,w_1)\to w_1 :p \leq c},\;\;  R(w_0,w_1)\to w_1 :p\lor q > c,\;\; R(w_0,w_1)\land w_1 : q \leq c,\;\;\D$} 
    \RightLabel{$\Box\lt$}
    \UIC{$w_0:  \Box(p\lor q) > c,\;\; \underline{w_0 : \Box p \leq  c},\;\; w_0: \Diam q  \leq  c,\;\;  c   < 1$} 
    \RightLabel{$\lor\lt$}
    \UIC{$w_0:  \Box(p\lor q) > c,\;\;\underline{w_0 : (\Box p\lor \Diam q) \leq c},\;\; c  < 1$}
    \RightLabel{$\to <$}
    \UIC{\underline{$w_0:  \Box(p\lor q) \to (\Box p\lor \Diam q)  < 1$}}
    \DP
  \end{array}
  \]  
  \vspace{-2ex}
  \caption{A saturated branch.}
  \label{fig:saturatedbranch}
\end{figure}

\begin{lemma}\label{lemma:count}
  Let $\Bcal=\stru{\G_0,\dots,\G_n}$ be a saturated branch and $\s$ a
  solution to $\Atmp{\G_n}$.  Let $\M=\Mod{\Atmp{\G_n},\s}$, $\Ical$ a
  canonical $\M$-interpretation  and $\chi\in
  \bigcup_{k\in\{0,\dots,n\}}\G_k$.  Then, $\M\models_\Ical \chi$.
\end{lemma}

\begin{proof}
  Since $\M=\Mod{\Atmp{\G_n},\s}$, with $\s$ is a solution to $\Atmp{\G_n}$,
  it holds that:
  \begin{enumerate}[label=(\arabic*), ref=(\arabic*),leftmargin=*]
  \item\label{lemma:count:atmp}
    $\M\models_\Ical \Atmp{\G_n}$.
  \end{enumerate}
  We prove  $\M\models_\Ical \chi$ by induction on $\wg{\chi}$;
  we assume $\M=\stru{W,R,e}$.
  Let $\wg{\chi}=0$; then $\chi$ is atomic, hence $\chi\in\G_n$,
  and $\M\models_\Ical \chi$ by~\ref{lemma:count:atmp}.
  
  Let $\wg{\chi}>0$ and let us assume that  $\chi\not\in \G_n$.
  Then, there exists $k \geq  0$ such that
  $\app{\chi}{\G_k}{\G_{k+1}}$. 
  We proceed by a case analysis on the structure of $\chi$.

  Let $\chi= w:\a\land \b\gt t$.  Since $\app{\chi}{\G_k}{\G_{k+1}}$,
  the applied rule is $\land\gt$, hence $w:\a\gt t\in \G_{k+1}$ and
  $w:\b \gt t\in \G_{k+1}$.  By the induction hypothesis
  $\M\models_\Ical w:\a\gt t$ and $\M\models_\Ical w:\b \gt t $, and
  this implies $\M\models_\Ical w:\a\land \b\gt t$.  The case $\chi=
  w:\a\lor \b\lt t$ is similar.

  Let $\chi= w:\a\land \b\lt t$.  Since $\app{\chi}{\G_k}{\G_{k+1}}$,
  the applied rule is $\land\lt$, hence $w:\a\lt t \in \G_{k+1}$ or
  $w:\b\lt t \in \G_{k+1}$.  By the induction hypothesis
  $\M\models_\Ical w:\a\lt t $ or $\M\models_\Ical w:\b\lt t $, hence
  $\M\models_\Ical w:\a\land \b\lt t$.  The case $\chi= w:\a\lor \b\gt
  t$ is similar.

  Let $\chi= w:\a\to \b < t$; we only consider the case
  $\b\not\in\PV\cup\{\bot\}$.  Since $\app{\chi}{\G_k}{\G_{k+1}}$, the
  applied rule is $\to<$, hence:
 \begin{enumerate}[label=(\Alph*), ref=(\Alph*),leftmargin=*]
 \item\label{lemma:count:impless:1}
   $\{\,w:\a > w : q,\, w:\b \leq w : q,\,w: q < t\,\} \subseteq \G_{k+1}$. 
 \end{enumerate}
 By the induction hypothesis,
 for every constraint $\chi'$  displayed in~\ref{lemma:count:impless:1},
 $\M\models_\Ical \chi'$, hence:
 \[
 e(w,\a) > e(w,q) \qquad e(w,\b) \leq e(w,q) \qquad e(w,q) < \Ical(t)
 \]  
 Since $e(w,\a) > e(w,\b)$, we get  $e(w,\a\to \b) = e(w,\b)$, hence  $e(w,\a\to \b) < \Ical(t)$,
 which implies $\M\models_\Ical  w:\a\to \b < t  $.

 Let $\chi= w:\a\to \b \leq t$; we only consider the case
 $\b\not\in\PV\cup\{\bot\}$.  Since $\app{\chi}{\G_k}{\G_{k+1}}$, the
 applied rule is $\to\leq$, hence one of the following two subcases
 holds:
 
 \begin{enumerate}[label=(B\arabic*), ref=(B\arabic*),leftmargin=*]
 \item\label{lemma:count:impleq:1}
   $t\geq 1\in \G_{k+1}$;
   
 \item\label{lemma:count:impleq:2}
   $\{\,w:\a> w :q,\,w:\b\leq w:q,\,w:q\leq t\,\} \subseteq \G_{k+1}$.
 \end{enumerate}
 Let us assume that~\ref{lemma:count:impleq:1} holds.  By the
 induction hypothesis, we get $\M\models_\Ical t\geq 1$, and this
 implies $\M\models_\Ical w:\a\to \b \leq t$.  Let us assume
 that~\ref{lemma:count:impleq:2} holds.  By the induction hypothesis,
 for every constraint $\chi'$ in~\ref{lemma:count:impleq:2},
 $\M\models_\Ical \chi'$, hence:
 \[
 e(w,\a) > e(w,q) \qquad e(w,\b) \leq e(w,q) \qquad e(w,q) \leq \Ical(t)
 \]  
 Since $e(w,\a) > e(w,\b)$, we get $e(w,\a\to \b) = e(w,\b)$, hence
 $e(w,\a\to \b) \leq \Ical(t)$, which implies $\M\models_\Ical w:\a\to
 \b \leq t$.

 Let $\chi= w:\a\to \b \gt t$; we only consider the case
 $\b\not\in\PV\cup\{\bot\}$.  Since $\app{\chi}{\G_k}{\G_{k+1}}$, the
 applied rule is $\to\gt$, hence one of the following two subcases
 holds:
 
 \begin{enumerate}[label=(C\arabic*), ref=(C\arabic*),leftmargin=*]
 \item\label{lemma:count:implgt:1}
   $\{\,w:\a \leq w:q,\,w:\b\geq  w:q,\,1\gt t\,\}\subseteq\G_{k+1}$;
   
 \item\label{lemma:count:implgt:2}
   $w:\b\gt t\in\G_{k+1}$.
 \end{enumerate}
 
 \noindent
 Assume~\ref{lemma:count:implgt:1}.  By the induction hypothesis, for
 every constraint $\chi'$ in~\ref{lemma:count:implgt:1},
 $\M\models_\Ical \chi'$, hence:
 \[
 e(w,\a) \leq  e(w,q) \qquad e(w,\b) \geq e(w,q) \qquad 1 \gt \Ical(t)
 \]  
 Since $e(w,\a)\leq e(w,\b)$, we get $e(w,\a\to \b)=1$, hence
 $e(w,\a\to \b)\gt \Ical(t)$, which implies $\M\models_\Ical\chi$.  In
 case~\ref{lemma:count:implgt:2}, by the induction hypothesis we get
 $\M\models_\Ical w:\b\gt t$, namely $e(w,\b)\gt\Ical(t)$.  Since
 $e(w,\a\to\b)\geq e(w,\b)$, we get $e(w,\a\to\b)\gt\Ical(t)$, hence
 $\M\models_\Ical w:\a\to \b \gt t$.  The cases concerning
 $R$-constraints can be proved as the corresponding propositional
 cases.

 Let $\chi= w:\Box\a\lt t$.  Since $\app{\chi}{\G_k}{\G_{k+1}}$, the
 applied rule is $\Box\lt$, hence:
 \begin{enumerate}[label=(\Alph*), ref=(\Alph*),leftmargin=*,start=4]
   
 \item\label{lemma:count:boxlt}
   $R(w,w_1)\to w_1:\a\lt t \, \in \G_{k+1}$.
 \end{enumerate}
 Let $r_1$ be the value of $R(w,w_1)\to e(w_1:\a)$.  By the induction
 hypothesis $\M\models_\Ical R(w,w_1)\to w_1:\a\lt t $, hence $r_1\lt\
 \Ical(t)$.  Since $e(w,\Box\a)\leq r_1$, we get $e(w,\Box\a) \lt
 \Ical(t)$, and this implies $\M\models_\Ical w:\Box\a\lt t$.  The
 case $\chi= w:\Diam\a\gt t$ is similar.

 It remains to consider the cases where $\wg{\chi}>0$ and $\chi\in
 \G_n$, thus $\chi= w:\Box\a\gt t$ or $\chi= w:\Diam\a\lt t$.  Let
 $\chi= w:\Box\a\gt t$; we show that:

 \begin{enumerate}[label=(\Alph*), ref=(\Alph*),leftmargin=*,start=5]
 \item\label{lemma:count:boxgt}
   $\M\models_\Ical R(w,w') \to w':\a\gt t$, for every  $w'\in W$.
 \end{enumerate}
 Let $w'\in W$; we show that $\M\models_\Ical R(w,w') \to w':\a\gt t$.
 Let us assume that $R(w,w')$ does not occur in $\G_n$.  By definition
 of $\M$, the value of $R(w,w')$ is 0, hence we have to show that
 $\M\models_\Ical 1 \gt t$.  The case where $\gt$ is $\geq$ is
 immediate.  Let $\gt$ be $>$.  Since $\chi\in\G_n$, we get $1
 >t\in\Atmp{\G_n}$, thus $\M\models_\Ical1 > t$
 by~\ref{lemma:count:atmp}.  Let us assume that $R(w,w')$ occurs in
 $\G_n$.  Then, there exists $k > 0$ such that $R(w,w')$ is
 introduced in $\G_{k}$ by an application of a modal rule $\rho$,
 namely:
 \begin{itemize}
 \item $\app{\chi'}{\G_{k-1}}{\G_k}$, with $\chi' = w :\Box \varphi \lt' t'$ 
   or  $\chi' = w :\Diam\varphi \gt' t'$,
   and the new label introduced by $\rho$ is  $w'$.
 \end{itemize}
 Since $\Bcal$ is saturated, $\G_{k-1}$ is plain, hence the non-atomic
 constraints in $\G_{k-1}$ are modal.  As a consequence, no rule
 applied in the sub-branch $\stru{\G_{k-1},\dots,\G_n}$ can introduce
 in the premise new constraints $u\abstractorder t'$ such that $u$ has
 label $w$; this implies that $\chi\in \G_{k-1}$.  By definition of
 $\rho$, we get $\Phibd{\G_{k-1},w,w'}\subseteq \G_{k}$, hence
 $R(w,w') \to w':\a\gt t \in\G_{k}$. By the induction hypothesis, we
 get $\M\models_\Ical R(w,w') \to w':\a\gt t$, and this concludes the
 proof of~\ref{lemma:count:boxgt}.  Let $e(w,\Box\a)=r$.  By the witnessing condition,
  there exists $w^\star\in W$ such that the value of
 $R(w,w^\star) \to e(w^\star,\a)$ is $r$.  By~\ref{lemma:count:boxgt},
 $r\gt \Ical(t)$, namely $e(w,\Box\a)\gt \Ical(t)$; we conclude
 $\M\models_\Ical w:\Box\a\gt t$.  The case $\chi= w:\Diam\a\lt t$ is
 similar.
\end{proof}

\noindent
As a consequence, we get:
\begin{proposition}\label{prop:satBranch}
  Let $\Bcal=\stru{\G_0,\dots,\G_n}$ be a saturated branch.
  \begin{enumerate}[label=(\roman*), ref=(\roman*),leftmargin=*]
  \item\label{prop:satBranch:1} For every solution  $\s$  to
    $\Atmp{\G_n}$, $\Mod{\Atmp{\G_n},\s}$ is a discrete
    countermodel for $\G_0$.
      
  \item\label{prop:satBranch:2} If $\G_0=\{ w : \varphi < 1\}$, then
    $\varphi\not\in\GWL$.
  \end{enumerate}
\end{proposition}

\noindent
Let $\Bs$ be a backward proof search strategy for $\Calcw$; we say
that $\Bs$ is \emph{plain} iff all the modal rule applications
performed by $\Bs$ are plain.

\begin{lemma}\label{lemma:search}
  Let $\Bs$ be a plain proof search strategy for $\Calcw$ and let
  $\G_0$ be a finite multiset of constraints of the form
  $w:\varphi\abstractorder t$.  If $\Bs$ fails to prove $\G_0$, then a
  discrete countermodel for $\G_0$ can be built.
\end{lemma}

\begin{proof}
  Assume that $\Bs$ fails.  By tracing the computation, we can build
  an open branch $\Bcal=\stru{\G_0,\dots,\G_n}$.  Since $\Bs$ is
  plain, the branch $\Bcal$ is saturated.  Let $\s$ be a solution to
  $\Atmp{\G_n}$; by Prop.~\ref{prop:satBranch}\ref{prop:satBranch:1},
  we get that $\Mod{\Atmp{\G_n},\s}$ is a discrete countermodel for
  $\G_0$.
\end{proof}

\noindent
To sum up, we get the completeness of $\Calcw$ and the finite model
property for $\GWL$.

\begin{proposition}[Completeness]\label{prop:compl}
  Let $\G$ be a finite multiset of constraints of the form
  $w:\varphi\abstractorder t$.
  \begin{enumerate}[label=(\roman*), ref=(\roman*),leftmargin=*]    
  \item\label{prop:compl:1}
    $\provesw{\Gamma}$ iff, for every \GWM-model  $\M$,  $\M\not\models\G$.

  \item\label{prop:compl2}
    $\varphi\in\GWL$ iff  $\provesw w:\varphi < 1$; moreover, if
    $\varphi\not\in\GWL$, 
    there exists a discrete countermodel for $\varphi$.
\end{enumerate}
\end{proposition}

\noindent
Let $\varphi \notin \GWL$ and let $\M$ be the discrete countermodel
extracted from an opennn branch with root $w_0:\varphi < 1$. One can
easily show that the depth of $\M$ is bounded by $\size{\varphi}$ and
that every world in $\M$ has at most $\size{\varphi}$ $R$-successors.
Consequently, the size of $\M$ is
$O(\size{\varphi}^{\size{\varphi}})$; moreover, we can define formulas
$\varphi$ having countermodels of size $2^{\size{\varphi}}$ (e.g.,
consider the translation $\varphi^*$ of a $\KL$-formula $\varphi$ such
that $\varphi$ has a countermodel of exponential size).  Note that
$\M$ cannot be constructed incrementally, one branch at a time, since
the constraints generated during the expansion of a branch have global
validity and must be preserved throughout the entire construction.
Nevertheless, we can define a PSPACE algorithm to check if such a
countermodel exists without effectively constructing it; this is
discussed in the proof of the next proposition.

\begin{proposition}
The validity problem for $\GWL$ is PSPACE-complete.  
\end{proposition}

\begin{proof}
  Since the modal logic $\KL$ can be embedded into $\GWL$
  (Prop.~\ref{prop:embedding}) and $\KL$ is
  PSPACE-complete~\cite{Ladner:77}, the validity problem for $\GWL$ is
  PSPACE-hard.  It remains to prove that $\GWL$-validity can be
  checked in PSPACE.  To this aim, we sketch a PSPACE procedure to
  decide the membership of a formula to $\GWL$.  Henceforth, we assume
  that the formula to be checked is $\varphi$.  As discussed above, if
  $\varphi\not\in\GWL$, then there exists a tree-like countermodel
  $\M=\stru{W,R,e}$ for $\varphi$ such that the number of worlds is
  $O(\size{\varphi}^{\size{\varphi}})$.  We can also bound the number
  $N$ of values needed to define $R$ and $e$: in the worst case, the
  values $e(w,p)$ and $R(w,w')$ must all be pairwise distinct, for
  every world $w$, $w'$ in $\M$ and every propositional variable $p$
  occurring in $\varphi$; accordingly, $N$ is in
  $O(\size{\varphi}^{\size{\varphi}})$.  We set
  $\Qcal=\{0\}\cup\{~1/k~|~1\leq k\leq N\}$; note that the elements of
  $\Qcal$ can be generated in space $O(\size{\varphi})$.  By the above
  discussion, if $\varphi\not\in\GWL$ there exists a countermodel $\M$
  such that the images of $R$ and $e$ are contained in $\Qcal$.  To
  define the procedure, we need a global map $\Phi$ that specifies the
  intended values of the subformulas of $\varphi$ at each world of the
  model under construction.  More precisely, given a world name $w$,
  $\Phi(w)$ is a set of assignments of the form $e(w,\a)=r$, where
  $\a$ is a subformula of $\varphi$ and $r\in\Qcal$. We remark that
  $\Phi(w)$ can be stored in space $O(|\varphi|)$; moreover, to
  implement the search procedure, we do no not need to store the whole
  map $\Phi$, but at most $O(\size{\varphi})$ entries of $\Phi$.  A
  $\GWM$-model $\M=\stru{W,R,e}$ complies with $\Phi$ iff, for every
  $w$ in $W$ and every assignment $e(w,\a)=r$ stored in $\Phi(w)$, the
  value of $e(w,\a)$ (in $\M$) is $r$.  The decision procedure is
  implemented by the recursive procedure $\Search$ having the
  following specification:

\begin{itemize}
\item Let $w$ be a world name and let us assume that $\Phi(w)\neq
  \emptyset$.  The procedure $\Search(w)$ returns $\Success$ if a
  $\GWM$-model complying with $\Phi$ can be built, $\Fail$ otherwise.
\end{itemize}

\noindent
The procedure, at some points, requires the selection of a value from
$\Qcal$, possibly satisfying some conditions: we call this a
$\Qcal$-choice.  $\Qcal$-choices can be revoked and redone; such
backtrack steps can be implemented in space $O(\size{\varphi})$.  The
computation of $\Search(w)$ consists of two main steps.

\begin{enumerate}[label*=(Step~\arabic*)]
\item\label{step:1} We apply propositional saturation reductions to
  $\Phi(w)$ as long as possible. At each step, an assignment
  $e(w,\a\star\b)=r$ (with $\star\in \{\,\land,\,\lor,\,\to\,\}$) in
  $\Phi(w)$ is removed and replaced by the assignments $e(w,\a)=r_1$
  and $e(w,\b)=r_2$ required to justifies the value $r$, where the
  values $r_1$ and $r_2$ are selected with $\Qcal$-choices.  For
  instance, let us assume that the assignment at hand is $e(w,\a\to\b)=r$.
  \begin{itemize}
  \item If $r=1$, we must have $e(w,\a) \leq e(w,\b)$; the values
    $r_1$ and $r_2$ are selected with $\Qcal$-choices so that $r_1\leq
    r_2$; in $\Phi(w)$, $e(w,\a\to \b)=r$ is replaced with $e(w,\a )=
    r_1$ and $e(w,\b)=r_2$.
  \item If $r<1$, a value $r_1$ is selected with a $\Qcal$-choice so
    that $r_1> r$; in $\Phi(w)$, $e(w,\a\to \b)=r$ is replaced with
    $e(w,\a)= r_1$ and $e(w,\b) = r$ (thus, $r_2$=r).
  \end{itemize}
  At the end of~\ref{step:1}, the formulas occurring in $\Phi(w)$ are
  either atomic or modal.
  
  \begin{itemize}
  \item If $\Phi(w)$ contains two evaluations $e(w,\a)=r_1$ and
    $e(w,\a)=r_2$ such that $r_1\neq r_2$, a backtrack step is needed.
    If no new $\Qcal$-choice is possible, the call $\Search(w)$
    returns $\Fail$.

  \item Otherwise, if $\Phi(w)$ contains only atomic constraints, the
    call $\Search(w)$ returns $\Success$.
    
  \item Otherwise, $\Phi(w)$ contains at least  a
    modal formula and the computation continues with~\ref{step:2}.
  \end{itemize}  

\item\label{step:2} Let $e(w,\#_1\a_1)=r_1,\dots,e(w,\#_n\a_n)=r_n$
  (with $\#_k\in\{\Box,\Diamond\}$) be all the assignments in $\Phi(w)$
  involving modal formulas.  We start a loop indexed by $k$ ranging
  over $\{1,\dots,n\}$.  Let $k$ be the current iteration and let us
  assume that the assignment of index $k$ has the form
  $e(w,\Box\a_k)=r_k$ (the case $e(w,\Diam\a_k)=r_k$ is similar).  We
  generate a new world $w_k$ and select with a $\Qcal$-choice a value
  $r'_k$ for $R(w,w_k)$ and a value $r''_k$ for $e(w_k,\a_k)$ such
  that $r'_k\to r''_k = r_k$. The world $w_k$ is intended to represent
  the world witnessing that $e(w,\Box\a_k)=r_k$. We set
  $\Phi(w_k)=\{e(w_k,\a_k)=r''_k\}$.  We expand $\Phi(w_k)$ by adding
  the assignments required to account for the assignments
  $e(w,\#_j\a_j)=r_j$ occurring in $\Phi(w)$, where $j\neq k$. More
  precisely:
  \begin{itemize}
  \item for every $e(w,\Box \a_j)=r_j$ in $\Phi(w)$ ($j\neq k$), we
    select with a $\Qcal$-choice a value $r'_j$ such that $r'_k \to
    r'_j \geq r_j$ and we add $e(w_k,\a_j)=r'_j$ to $\Phi(w_k)$;

  \item for every $e(w,\Diam \a_j)=r_j$ in $\Phi$ ($j\neq k$), we
    select with a $\Qcal$-choice a value $r'_j$ such that $r'_k \land
    r'_j \leq r_j$ and we add $e(w_k,\a_j)=r'_j$ to $\Phi(w_k)$.
  \end{itemize}
  We compute recursively $\Search(w_k)$; notice that the modal
  operators occurring in $\Phi(w_k)$ are fewer than the modal
  operators in $\Phi(w)$, thus the recursion is terminating. We also
  stress that there are at most $\size{\phi}$ nested recursive calls.
  Let us assume that $\Search(w_k)$ returns $\Success$.  If $k=n$
  (last iteration), the call $\Search(w)$ returns $\Success$.
  Otherwise the computation continues with next iteration $k+1$.  We
  point out that the entries of the map $\Phi$ exploited in the
  computation of $\Search(w_k)$ are no longer required, so they can be
  removed from the memory.  Let us assume that $\Search(w_k)$ returns
  $\Fail$.  Then, a backtrack step is needed.  If no new
  $\Qcal$-choice is possible, the call $\Search(w)$ returns $\Fail$.
\end{enumerate}
The procedure can be implemented so that $\Search(w)$ is computed in
space $O(\size{\varphi})$.

We can exploit $\Search$ to decide the validity of $\varphi$ in
$\GWL$.  We initialize $\Phi$ by setting $\Phi(w_0)=\{e(w_0,\varphi)=
r_0\}$, where $r_0$ is any value in $\Qcal$ different from 1.  If
$\Search(w_0)$ returns $\Success$, then a model $\M$ containing a
world $w_0$ such that $e(w_0,\varphi)=r_0$ can be built; since $r_0 <
1$, we conclude $\varphi\not\in \GWL$.  If $\Search(w_0)$ returns
$\Fail$, for every possible choice of $r_0$, it follows that no
$\GWM$-countermodel for $\varphi$ can be built; accordingly,
$\varphi\in \GWL$.
\end{proof}

\section{Conclusions and Future Work}

We have implemented a proof-search procedure for the calculus $\Calcw$
that supports countermodel extraction.  The prover, called
\texttt{gwref}~\cite{gwrefProver}, performs a standard backward
depth-first proof search, relying on the JTabWb
engine~\cite{FerrariFF:17a}. The consistency of atomic constraints is
verified using the Choco-solver Java library~\cite{ChocoSolver:2025}.
We emphasize that the calculus $\Calcw$ and the associated
proof-search strategy can be smoothly adapted to treat the logic
$\GWCL$~\cite{FerFioRod:2025}, characterized by crisp $\GWM$-models,
by refining the notion of consistency for a set of atomic constraints
$\Gat$: a solution $\s$ to $\Gat$ must satisfy $\s(R(w,w')) \in
\{0,1\}$ for every pair of labels $w$, $w'$.  To the best of our
knowledge, the only other prover available for modal fuzzy logics is
mNiBLoS~\cite{Vidal:16}, an SMT-based solver designed for continuous
t-norm-based logics, including G\"odel--Dummett Logic.  However,
neither $\GWL$ nor $\GWCL$ are supported by mNiBLoS.


The results presented in this paper are highly relevant to the study
of intermediate predicate logics. It is well known that modal logics
can be used to characterize relevant fragments of intermediate
first-order logics, particularly certain one-variable fragments. From
this perspective, the witnessed semantics, together with the consequent
finite model property established here, point to deep connections with
proof-theoretic properties of intermediate predicate logics, such as
Herbrand’s Theorem, interpolation, and alternative forms of
Skolemization investigated in~\cite{BI-IGPL2016} for broad classes of
intermediate intuitionistic predicate logics with the finite model
property.  Motivated by these connections, and with the aim of
situating $\GWL$ within the broader landscape of intuitionistic modal
logics, we plan to develop a birelational Kripke semantics for $\GWL$,
extending to this setting the approach introduced for $\GWCL$
in~\cite{FerFioRod:2025}.

We also plan to extend the procedure to other witnessed G\"odel modal
logics that impose additional constraints on the accessibility
relation (e.g., reflexivity, transitivity, seriality).  In contrast,
treating G\"odel modal logics with non-witnessed models seems to be
considerably more challenging.

\subsubsection*{Acknowledgments}

This project has received funding from the European Union's Horizon
2020 research and innovation programme under the Marie
Sk{\l}odowska-Curie grant agreement No 101007627 (MOSAIC project) and
from the European Union's Horizon Europe research and innovation
programme under the Marie Sk{\l}odowska-Curie grant agreement
No. 101299559 (SemPER project).

Camillo Fiorentini is member of the \emph{Gruppo Nazionale Calcolo
  Scientifico - Istituto Nazionale di Alta Matematica (GNCS-INdAM)}.
Mauro Ferrari and Paolo Giardini are members of the \emph{Gruppo
  Nazionale per le Strutture Algebriche, Geometriche e le loro
  Applicazioni - Istituto Nazionale di Alta Matematica
  (GNSAGA-INdAM)}.

The last author would like to remark that this work was completed
despite the lack of support from the scientific funding agencies of
Argentina.


\end{document}